\documentclass{article}
\usepackage{ijcai22}

\usepackage{times}
\usepackage[inline]{enumitem}
\usepackage{url}
\usepackage[hidelinks]{hyperref}
\usepackage[utf8]{inputenc}
\usepackage[small]{caption}
\usepackage{graphicx}
\usepackage{amsmath}
\usepackage{amsthm}
\usepackage{booktabs}
\newtheorem{theorem}{Theorem}
\newtheorem{definition}{Definition}

\newtheorem{lemma}{Lemma}

\usepackage{amssymb}
\usepackage{thmtools}
\usepackage{mathtools}
\usepackage{autonum}
\usepackage[capitalise]{cleveref}
\usepackage[all]{foreign}
\usepackage{forest}
\usepackage{multicol}
\usepackage{lipsum}
\usepackage{algorithm}
\usepackage{algpseudocode}

\usepackage{packages/commands} 
\usepackage{packages/graphics} 
\usepackage{packages/logic}    
\usepackage{packages/symbols}  

\usepackage{tikz}
\usepackage{packages/bpmn-events}
\usepackage{packages/bpmn-gateways}
\usepackage{sparklines}

\allowdisplaybreaks

\newif\ifblind
\blindfalse

\newif\ifappendix
\appendixfalse

\IfFileExists{arxiv.tex}{

\blindfalse
\appendixtrue}

\ifappendix
  \title{Linear Temporal Logic Modulo Theories over Finite Traces \\
  (Extended Version)}
\else
  \title{Linear Temporal Logic Modulo Theories over Finite Traces}
\fi

\ifblind
  \author{Submission \#2087}
\else
  \author{
    Luca Geatti \and
    Alessandro Gianola \and
    Nicola Gigante
    \affiliations
    Free University of Bozen-Bolzano, Italy
    \emails
    \{geatti,gianola,gigante\}@inf.unibz.it
  }
\fi

\begin{document}

\maketitle

\begin{abstract}
  This paper studies Linear Temporal Logic over Finite Traces (\LTLf) where
  proposition letters are replaced with first-order formulas interpreted over
  arbitrary theories, in the spirit of Satisfiability Modulo Theories. The
  resulting logic, called \LTLf Modulo Theories (\LTLfMT), is semi-decidable.
  Nevertheless, its high expressiveness comes useful in a number of use cases,
  such as model-checking of data-aware processes and data-aware planning.
  Despite the general undecidability of these problems, being able to solve
  satisfiable instances is a compromise worth studying. After motivating and
  describing such use cases, we provide a sound and complete semi-decision
  procedure for \LTLfMT based on the SMT encoding of a one-pass tree-shaped
  tableau system.
  The algorithm is implemented in the \BLACK satisfiability checking tool, and
  an experimental evaluation shows the feasibility of the approach on novel
  benchmarks.
\end{abstract}


\section{Introduction}
\label{sec:introduction}

Linear Temporal Logic (\LTL)~\cite{pnueli1977temporal} and \LTL over finite
traces (\LTLf)~\cite{DeGiacomoV13} are common languages to express temporal
properties in the fields of \emph{formal verification} and AI. In particular,
\LTLf has recently gained traction in AI and \emph{business process
modeling}~\cite{DeGiacomoMGMM14}. In these fields, reasoning over finite traces
is more natural.

However, in the modeling of \emph{data-aware processes} \cite{MSCS20,LiDV17},
the propositional nature of \LTLf is a severe limitation.
These are systems whose behavior depends on and/or manipulate a persistent data
storage such as a relational database. In these contexts, one would like to
express actions and constraints that depend on the contents of the database,
which is an inherently first-order object. A similar use case is that of
\emph{data-aware} planning, \ie planning problems~\cite{GhallabNT04} where
conditions of actions depend on the contents of the persistent data storage.

In order to obtain the expressiveness needed to model and reason about these
scenarios, this paper introduces and studies \emph{\LTLf Modulo Theory}
(\LTLfMT), a logic that extends \LTLf by replacing propositional symbols with
first-order formulas over an arbitrary theory, in the spirit of
\emph{Satisfiability Modulo Theory} (SMT). While \LTLfMT is easily seen to be
undecidable in general, for decidable theories it is \emph{semi-decidable}, \ie
a positive answer can always be obtained for satisfiable instances, while
reasoning might not terminate over unsatisfiable instances. We argue that, for
complex problems and scenarios like those mentioned above, where reasoning is
inherently undecidable anyway, being able to solve satisfiable instances is a
compromise worth studying. 

In particular, we provide a semi-decision procedure for \LTLfMT satisfiability
based on the one-pass tree-shaped tableau for \LTL by
Reynolds~[\citeyear{Reynolds16a}]. In contrast to classic graph-shaped
tableaux for \LTL, Reynolds' builds a tree structure where a single independent
pass is sufficient to decide whether to accept or reject the current branch.
Here, we adapt Reynolds' tableau to \LTLfMT, showing that, for decidable
underlying first-order theories, \LTLfMT is semi-decidable. While  
first-order extensions of \LTL have already been studied before, our setting is
more specific: on the one hand, we are more general than, \eg
\cite{cimatti2020smt} by supporting quantified formulas. On the other hand, we
restrict general first-order \LTL~\cite{KontchakovLWZ04,ArtaleMO19} by not
allowing temporal operators inside quantifiers.\fitpar

Our main contribution is an SMT-based algorithmic technique for satisfiability
checking of such an expressive logic, which we implement in the \BLACK
tool~\cite{DBLP:conf/tableaux/GeattiGM19,GeattiGMV21}. \BLACK is a recent
state-of-the-art satisfiability checker for \LTL and \LTLf based on a SAT
encoding of Reynolds' tableau. Inspired by that, we provide an SMT encoding of
our tableau for \LTLfMT: suitably encoded SMT formulas represent the branches of
the tableau up to depth $k$, for increasing values of $k$.

An experimental evaluation assesses the applicability of our approach. With a
number of novel benchmarks over different underlying theories, we show that
\BLACK is able to reason over satisfiable instances and, in some cases, over
unsatisfiable instances as well, with promising performance. 

The paper is organized as follows. We introduce \LTLfMT in \cref{sec:ltlfmt}.
Then, \cref{sec:use-cases} discusses some relevant use cases that motivate our
work. \Cref{sec:tableau,sec:encoding} provide a semi-decision procedure for
\LTLfMT and its SMT encoding. \Cref{sec:experiments} evaluates the
implementation of this procedure in the \BLACK satisfiability checker, and
\cref{sec:conclusions} concludes.
\ifappendix
Proofs are found in the Appendix.
\else
Proofs can be found in \cite{arxiv-ltlfmt}.
\fi


\section{LTL Modulo Theories over Finite Traces}
\label{sec:ltlfmt}

$\LTLfMT$ extends Linear Temporal Logic over Finite Traces (\LTLf) by replacing
proposition letters with first-order formulas over arbitrary theories, similarly
to how Satisfiability Modulo Theory (SMT) extends the classical Boolean
satisfiability problem. Before going into the formal details of syntax and
semantics, let us show some examples.

Consider the theory of \emph{linear integer arithmetic} (LIA). This first-order
theory predicates over the set of integer numbers $\Z$, and interprets the $+$
function symbol and the $\le$ relation symbol as the usual sum and comparison
between integer numbers. There is no multiplication symbol. With \LTLfMT over
LIA, we can express formulas such as the following:
\begin{align}
  &\ltl{G(x = y + y)}
  \quad \ltl{(x < y) U y = 0}
  \quad \ltl{G(x > 5) & F(x < 0)}
\end{align}
This first formula on the left states that the interpretation of the variable
$x$ must always (\ie in any time point) be the double of the variable $y$. The
second formula states that the variable $x$ must remain less than $y$ until the
first time $y$ becomes equal to zero. The third one, which is unsatisfiable,
demands $x$ to be always greater than $5$ but also to become less than $0$ in
the future.  Quantifiers are allowed, but \emph{temporal operators} cannot be
nested inside them. For example, the following formula states that $x$ has
always to be even:
\begin{equation}
  \ltl{G(\exists y (x = y + y))}
\end{equation}
While these examples are interesting, there is often the need to go further, by
relating the value of variables at a given time point to their value at the
following instant. 
This is expressed by two term constructors, $\nextvar x$ and
$\wnextvar x$, which both represent, in a different way, the value of the
variable $x$ at the next state. The difference lies in how these terms behave
at the end of the finite trace. When a first-order formula contains some $\nextvar
x$, the next state is required to exist, while it is not required to exist if
it contains only $\wnextvar x$ terms. This replicates the difference between
the \emph{tomorrow} and \emph{weak tomorrow} temporal operators in
\LTLf~\cite{DeGiacomoV13,DeGiacomoDM14}. With these two term constructors in
place, we can express more interesting things such as:
\begin{align}
  &\ltl{x = 0 & ((\nextvar x = x + 1) U x = 42)} \\
  &\ltl{x = 0 & G(\wnextvar x > x & \exists y(x = y + y))} \\
  &\ltl{y = 1 & G(\wnextvar y = y+1 & x = y+y)}
\end{align}
The first formula states that variable $x$ behaves as a counter that increments
until reaching $x=42$. The second one makes $x$ represent \emph{any} strictly
increasing sequence of even numbers, while in the third formula $x$ takes
\emph{all} the first $n$ even numbers for any $n\ge 1$. Note that, in the second
and third formula, replacing $\wnextvar x$ with $\nextvar x$ would result in an
unsatisfiable formula, because each time point would require the existence of
the next one, which is impossible in a finite-trace semantics.

\paragraph{Syntax}

We are now ready to delve into the details of syntax and semantics of \LTLfMT.
Let us start from the syntax. We work with a multi-sorted first-order signature
$\Sigma = \mathcal{S} \cup \mathcal{P} \cup \mathcal{C} \cup \mathcal{F} \cup
\mathcal{V} \cup \mathcal{W}$, composed of a set of sort symbols $\mathcal{S}$,
a set of predicate symbols $\mathcal{P}$, a set of constant symbols
$\mathcal{C}$, a set of function symbols $\mathcal{F}$, a set of variable
symbols $\mathcal{V}$, and a set of \emph{quantified} variables $\mathcal{W}$.
Each constant in $\mathcal C$ and each variable in $\mathcal V$ and $\mathcal W$
is associated with a sort symbol $S\in\mathcal S$, and so are the domains of
each relation symbol and the domains and ranges of each function symbol.

A \emph{$\Sigma$-term} $t$ is generated by the following grammar:
\begin{equation}
  t \coloneqq \ x \choice y \choice c \choice f(t_1,\dots,t_k)
  \choice \nextvar x \choice \wnextvar x
\end{equation}
where $x \in \mathcal{V}$ and $y\in\mathcal W$ are variables, $c \in
\mathcal{C}$ is a constant symbol, $f \in \mathcal{F}$ is a function symbol of
arity $k$, each $t_i$ (for any $i \in \set{1,\dots,k}$) is a $\Sigma$-term, and
$\nextvar$ and $\wnextvar$ are the \emph{next} and \emph{weak next}
constructors.
Note that the distinction between variables in $\mathcal V$ and in $\mathcal W$
is needed since it does not make sense to apply $\nextvar x$ and $\wnextvar x$
to a quantified variable (recall that $\nextvar x$ represents the value of $x$
at the next state). The grammar of \LTLfMT formulas over $\Sigma$ is the
following:
\begin{align}
  \alpha \coloneqq {} & p(t_1,\ldots,t_k) \\
  \lambda \coloneqq {} & \alpha \choice \neg\alpha  \
  \choice \lambda_1\lor\lambda_2 \choice \lambda_1\land\lambda_2
  \choice \exists x \lambda \choice \forall x \lambda \\
  \phi \coloneqq {} &
    \!\top \choice \lambda \choice
     \phi_1\! \lor\! \phi_2 \choice
     \phi_1\! \land\! \phi_2 \choice 
     \ltl{X \phi} \choice
     \ltl{wX \phi} \choice
     \ltl{\phi_1 U \phi_2} \choice
     \ltl{\phi_1 R \phi_2}
\end{align}
where $x\in\mathcal W$, $p \in \Sigma$ is an $k$-ary predicate symbol, each
$t_i$ is a $\Sigma$-term, and $\ltl{X}$, $\ltl{wX}$, $\ltl{U}$, and $\ltl{R}$
are the \emph{tomorrow}, \emph{weak tomorrow}, \emph{until}, and \emph{release}
temporal operators, respectively. In the definition of $\phi$, we force
$\lambda$ to not have free variables from $\mathcal W$. Formulas of type
$\lambda$ as defined above are called \emph{first-order} formulas. We assume the
usual shortcuts for temporal operators, such as $\ltl{F\phi\equiv\top U \phi}$
and $\ltl{G\phi\equiv\false R \phi}$. 
Note that by the grammar above, \LTLfMT formulas are always in \emph{negated
normal form}, for ease of exposition. The formulas of \LTLfMT are assumed to be
well-typed with regard to the sorts of all the involved symbols. One may also
consider \emph{past operators}, but we omit them here to ease exposition.

\paragraph{Semantics}
We use the standard notion of \emph{first-order $\Sigma$-structure} $M$ over the
first-order (multi-sorted) signature $\Sigma$, which consists of a domain and of
an interpretation $s^M$ of all sort, predicate, constant, and function symbols
$s \in \Sigma$. In particular, sort symbols are interpreted as pairwise disjoint
sets, whose union $\dom(M)$ is the \emph{domain} of $M$. In line with
\cite{SMT}, we define a \emph{theory} \theory as a (finite or infinite) class of
$\Sigma$-structures. 

Let $\Sigma = \mathcal{S} \cup \mathcal{P} \cup \mathcal{C} \cup \mathcal{F}
\cup \mathcal{V}\cup\mathcal{W}$ be a signature and let \theory be a theory. A
\emph{\theory-state} $s = \pair{M,\mu}$ is a pair made of a $\Sigma$-structure
$M \in \theory$ and a variable evaluation function $\mu : \mathcal{V} \to
\dom(M)$ assigning to each variable in $\mathcal{V}$ a value in the domain of
$M$. 

A word $\sigma = \seq{(M,\mu_0), \dots, (M,\mu_{n-1})}$ over the theory \theory
is a finite sequence of \theory-states over the same first-order structure.  We
define $\abs{\sigma} = n$.  It is worth noticing that any two states can differ
only in the variable evaluation functions, and \emph{not} in their domain nor in
the interpretations of the other symbols, which are rigidly interpreted. In the
context of first-order \LTL, this is called \emph{constant domain semantics}
\cite{hodkinson2000decidable}.

Given a term $t$, a word $\sigma=\seq{(M,\mu_0),\ldots,(M,\mu_{n-1})}$, an
integer $0\le i < n$, and a variable evaluation function $\xi:\mathcal W\to
\dom(M)$ for the variables in $\mathcal{W}$, the \emph{evaluation} of $t$ at
the instant $i$ on the trace $\sigma$ with \emph{environment} $\xi$, denoted
$\eval{t}_{\sigma,\xi}^i$, is the following:
\begin{enumerate}
  \item $\eval{x}_{\sigma,\xi}^i=\begin{cases}
    \mu_i(x) & \text{if $x\in\mathcal V$} \\
    \xi(x) & \text{if $x\in\mathcal W$}
  \end{cases}$
  \item $\eval{\nextvar x}^i_{\sigma,\xi} = 
    \eval{\wnextvar x}^i_{\sigma,\xi} = \mu_{i+1}(x)$
  \item $\eval{c}_{\sigma,\xi}^i=c^M$
  \item $\eval{f(t_1,\ldots,t_k)}_{\sigma,\xi}^i=
    f^M(\eval{t_1}_{\sigma,\xi}^i,\ldots,\eval{t_k}_{\sigma,\xi}^i)$
\end{enumerate}
Intuitively, when evaluating a term, the free variables from $\mathcal V$ are
evaluated according to the word $\sigma$, while the bound variables from
$\mathcal W$ are evaluated according to the environment~$\xi$. Note that
$\eval{t}_{\sigma,\xi}^i$ is \emph{well-defined} for $\sigma$ only if
\begin{enumerate*}[label=(\arabic*)]
  \item $i<\abs{\sigma}-1$ or 
  \item $i=\abs{\sigma}-1$ and  $t$ does not contain terms of type $\nextvar x$ or $\wnextvar x$. 
\end{enumerate*}
That is, we cannot evaluate a variable beyond the end of the word.
A first-order formula $\lambda$ is well-defined for $\sigma$ only if
$\eval{t}_{\sigma,\xi}$ is well-defined for all the terms $t$ appearing in $\psi$.
Given a variable evaluation function $\xi:\mathcal W\to\dom(M)$, we denote as
$\xi[x\leftarrow v]$ the function that agrees with $\xi$ except that
$\xi(x)=v$.

Given a theory \theory, the \emph{satisfaction modulo \theory} of a
\emph{first-order} formula $\psi$ over the word $\sigma$ at time point $i \in \N$
with environment $\xi$, denoted with $\sigma,\xi,i \models \psi$, is
inductively defined as:\fitpar
\begin{enumerate}
  \item $\sigma,\xi,i\models p(t_1,\ldots,t_k)$ is defined depending on whether
    terms of type $\nextvar x$ appear in $t_1,\ldots,t_k$:
  \begin{enumerate}
    \item if $\nextvar x$ appear in $t_1,\ldots,t_k$ for some variable $x$, then
          $\sigma,\xi,i\models p(t_1,\ldots,t_k)$ iff
          $\eval{t_1}^i_{\sigma,\xi},\ldots,\eval{t_k}^i_{\sigma,\xi}$ are
          well-defined and
          $(\eval{t_1}^i_{\sigma,\xi},\ldots,\eval{t_k}^i_{\sigma,\xi})\in p^M$;
    \item otherwise, $\sigma,\xi,i\models p(t_1,\ldots,t_k)$ if and only if at least
          one in $\eval{t_1}^i_{\sigma,\xi},\ldots,\eval{t_k}^i_{\sigma,\xi}$ is
          not well-defined or
          $(\eval{t_1}^i_{\sigma,\xi},\ldots,\eval{t_k}^i_{\sigma,\xi})\in p^M$;
  \end{enumerate}
  \item $\sigma,\xi,i \models \neg p(t_1,\ldots,t_k)$ iff 
    $\sigma,\xi,i \not\models p(t_1,\ldots,t_k)$;
  \item $\sigma,\xi,i \models \lambda_1 \lor \lambda_2$ iff 
    $\sigma,\xi,i\models\lambda_1$ or $\sigma,\xi,i\models\lambda_2$;
  \item $\sigma,\xi,i \models \lambda_1 \land \lambda_2$ iff 
    $\sigma,\xi,i\models\lambda_1$ and $\sigma,\xi,i\models\lambda_2$;
  \item $\sigma,\xi,i \models \exists x.\psi$ iff there exists a value
    $v\in\dom(M)$ such that $\sigma,\xi[x\leftarrow v],i\models \psi$;
  \item $\sigma,\xi,i \models \forall x.\psi$ iff for all values
    $v\in\dom(M)$ it holds that $\sigma,\xi[x\leftarrow v],i\models \psi$;
\end{enumerate}

Note that Items~1a and 1b above are those where the semantics of $\nextvar x$
and $\wnextvar x$ terms differ. In the first case, the next state must exist for
the atom to hold. In the second case, the atom holds by definition if the next
state does not exist. Given a theory \theory, the \emph{satisfaction modulo
\theory} of an \LTLfMT formula $\phi$ over the word $\sigma$ at time point
$i\in\N$, denoted as $\sigma,i\models\phi$, is inductively defined as follows:

\begin{enumerate}
  \item $\sigma,i\models \lambda$, where $\lambda$ is a first-order formula, if 
    $\sigma,\xi,i\models \lambda$ for some $\xi$;
  \item $\sigma,i \models \ltl{\phi_1 || \phi_2}$ iff $\sigma,i \models \phi_1$
    or $\sigma, i \models \phi_2$;
  \item $\sigma,i \models \ltl{\phi_1 && \phi_2}$ iff $\sigma,i \models \phi_1$
    and $\sigma, i \models \phi_2$;
  \item $\sigma,i \models \ltl{X \phi}$ iff $i<\abs{\sigma}-1$ and $\sigma,i+1
  \models \phi$;
  \item $\sigma,i\models \ltl{wX \phi}$ iff $i=\abs{\sigma}-1$ or
  $\sigma,i+1\models\phi$;
  \item $\sigma,i \models \ltl{\phi_1 U \phi_2}$ iff there exists a $j \ge i$
    such that $\sigma,j \models \phi_2$ and $\sigma,k \models \phi_1$ for all
    $i \le k < j$;
  \item $\sigma,i\models \ltl{\phi_1 R \phi_2}$ iff either
    $\sigma,j\models\phi_2$ for all $i\le j<\abs{\sigma}$, or there exists $k\ge
    i$ such that $\sigma,k\models\phi_1$ and $\sigma,j\models\phi_2$ for all
    $i\le j\le k$.
\end{enumerate}

We say that $\sigma$ \emph{satisfies $\phi$ modulo \theory} iff
$\sigma,0\models\phi$. The \emph{language modulo \theory} of $\phi$ is the set
of words $\sigma$ over \theory that satisfy $\phi$. A formula $\phi$ with no
temporal operators and no $\nextvar x$ or $\wnextvar x$ terms is a \emph{purely
first-order formula}. If a word $\sigma$ satisfies such a $\phi$, only the first
state is involved in its satisfaction. In this case we can write $s\models\phi$.
An atom $p(t_1,\ldots,t_k)$ is \emph{strong} if it contains at least a $\nextvar
x$ term. It is \emph{weak} if it contains $\wnextvar x$ terms but no $\nextvar
y$ terms. 



The satisfiability checking problem for \LTLfMT is easily seen to be
undecidable, depending on the theory. For example, with the LIA theory one can
easily encode the \textsc{PlanEx}-$(\mathcal{C}_{p+}, \mathcal{C}_\emptyset,
\mathcal{E}_{+1})$ \emph{numeric planning} problem, proved to be undecidable by
Helmert~\cite{Helmert02}. However, for suitable theories (see \cref{sec:tableau}), it is semi-decidable.

\paragraph{Expressiveness}

It is interesting at this point to wonder how much expressiveness we gain by
extending the classical (propositional) \LTLf \cite{DeGiacomoV13,DeGiacomoDM14}
in the way we previously shown. As \LTLf is a propositional logic, it is
trivially subsumed by \LTLfMT with a theory $\mathcal B$ consisting only of the
equality relation and a Boolean domain ($\{0,1\}$). One may wonder, however,
how expressive \LTLfMT over other theories is with regards to $\mathcal B$ (\ie
to \LTLf), when we abstract its models to propositional finite words. More
formally, for a signature $\Sigma = \mathcal{S} \cup \mathcal{P} \cup
\mathcal{C} \cup \mathcal{F} \cup \mathcal{V}\cup\mathcal{W}$ we consider
\emph{unary} predicates from $\mathcal P$ as propositional symbols, and given
a word $\sigma=\seq{s_0,\ldots,s_{n-1}}$ we define its \emph{Boolean
abstraction} $\mathbb B(\sigma)=\seq{\mathbb B(s_0),\ldots,\mathbb B(s_{n-1})}$
as follows: $P\in\mathbb B(s_i)$ iff $s_i\models \exists x P(x)$.
For any theory \theory, the Boolean abstraction of a language modulo \theory is
the set of the Boolean abstractions of all the words of the language. We can
prove the following.
\begin{theorem}
  There are languages definable in \LTLfMT whose Boolean abstraction cannot be 
  defined in \LTLf.
\end{theorem}

\begin{proof}
Consider the theory of Linear Integer Arithmetic (LIA) together
with uninterpreted unary predicates. It is well-known \cite{wolper1983temporal}
that \LTLf cannot express the language made of words where a proposition $p$
appears in at least all \emph{even} positions. Instead, such a language can be
defined as the Boolean abstraction of the language modulo LIA recognized by the
following \LTLfMT formula:
\begin{equation}
  x=0 \land \ltl{G( \wnextvar x = x+1 \land
    (\exists y(x=y+y) \to P(x)))}\tag*{\qedhere}
\end{equation}
\end{proof}

Note the essential role of the $\wnextvar x$ term here. Without such terms, the
satisfaction of first-order formulas in different time steps of a \LTLfMT model
would be completely unrelated, which would allow us to abstract them into
proposition symbols. As a matter of fact, we can prove the following.
\begin{restatable}{theorem}{booleanabstraction}
  The Boolean abstraction of any language definable in \LTLfMT \emph{without}
  $\nextvar x$ nor $\wnextvar x$ terms can also be defined in \LTLf, and
  \viceversa.
\end{restatable}

Regarding $\nextvar x$ and $\wnextvar x$ terms, it should be noted that the
\LTLfMT syntax only allows these term constructors to be applied to single
variables. However, the syntax may be extended to arbitrary nesting of
$\nextvar$ and $\wnextvar$ operators with a simple translation maintaining the
equisatisfiability. For instance, the formula $\ltl{x=1 & \nextvar x = 1 &
G(\wnextvar\wnextvar x = \wnextvar x + x)}$, expressing the Fibonacci sequence,
is equisatisfiable to $\ltl{x = 1 & \nextvar x = 1 & y = 1 & G(y = \wnextvar x &
\wnextvar y = y + x)}$.


\newcommand{\EUF}{\ensuremath{\mathcal{EUF}}\xspace}

\section{Use cases}
\label{sec:use-cases}


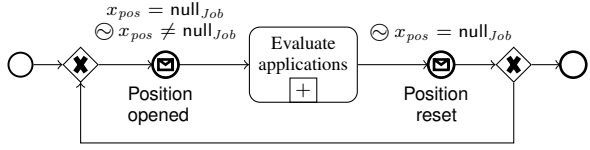
\begin{figure}
  \centering
  \begin{tikzpicture}[xscale = 0.55]
    \node[EndEvent] (se) at (0,0) {};
    
    \path (se.east) ++(0.7,0) node[anchor=west,ExclusiveGateway,draw] (xor1) {};

    \path (xor1.east) ++(1.3,0) node[anchor=west,MessageStartEvent,draw] (msg) 
    {};

    \path (msg.north) node[anchor=south,font=\scriptsize] 
      {$\begin{array}{c}
        x_{\mathit{pos}}=\mathsf{null}_{\mathit{Job}}\\
        \nextvar\kern-1.35em\sim x_{\mathit{pos}} \neq \mathsf{null}_{\mathit{Job}}
        \end{array}
      $};

    \path (msg.south) node[anchor=north,font=\scriptsize\sffamily,text width=1cm] {Position\\opened};

    \path (msg.east) ++(1.7,0) node[draw,anchor=west,font=\scriptsize,minimum height=1cm, text width=1.2cm, rounded corners] (eval) 
    { };
    \path (eval.north) node[anchor=north,font=\scriptsize,align=center,text width=1.2cm] {Evaluate\\applications};
    
    \node[draw,anchor=south,font=\scriptsize, inner sep=2pt] at (eval.south) { $+$ };

    \path (eval.east) ++(1.7,0) node[MessageStartEvent,draw,anchor=west] (undef) {};
    
      \path (undef.south) node[anchor=north,font=\scriptsize\sffamily,text width=1cm] {Position\\\,\,\,\,reset};

    \path (undef.north) node[anchor=south,font=\scriptsize] 
      {$\nextvar\kern-1.35em\sim x_{\mathit{pos}}=\mathsf{null}_{\mathit{Job}}$};

    \path (undef.east) ++(1,0) 
      node[ExclusiveGateway,draw,anchor=west] (xor2) {};
    
    \path (xor2.east) ++(1,0) node[StartEvent] (ee) {};

    \draw[->] (se) -- (xor1);
    \draw[->] (xor1) -- (msg);
    \draw[->] (msg) -- (eval);
    \draw[->] (eval) -- (undef);
    \draw[->] (undef) -- (xor2);
    \draw[->] (xor2) -- (ee);

    \draw[->] (xor2.south) -- ++(0,-0.8) -| (xor1.south) -- (xor1.south);
              
  \end{tikzpicture}
  \caption{Example business process model of a job application}\label{fig:bpmn}
\end{figure}

\paragraph{Verification of data-aware processes}
We take inspiration from verification of data-aware processes  to present an
interesting class of case studies. Specifically, we rely on the framework by
Calvanese \etal~[\citeyear{MSCS20,JAR20}], who introduced a general model of
transition systems interacting with relational databases, with the distinctive
feature that the transitions can query the content of the persistent data
storage. As shown there, the theory of \emph{equality and uninterpreted
functions} (EUF) provides an algebraic formalization of relational databases
with primary and foreign key dependencies: in particular, unary functions map
the primary key attribute of a relation schema to another attribute of the same
relation, and implicitly represent key dependencies. If EUF is combined with an
arithmetical theory such as LRA (e.g., following \cite{ijcar20ext,JAR22}), complex
datatype values can be injected into the database and used to model
non-deterministic inputs from external users \cite{BPM19}. A symbolic transition
system $\mathit{Sys}$ operating over databases can be formalized by means of:
(i) a sequence of variable symbols $x_1,\dots,x_n$ from $\mathcal{V}$,
describing the current configuration; (ii) a purely first-order formula
$I(x_1,\dots,x_n)$, describing the initial states; 
(iii) a weak first-order
formula $Tr(x_1,\dots, x_n)$ describing the transitions from the current
configuration to the next one. Given an \LTLfMT formula $\psi(x_1,\dots, x_n)$,
called \emph{property}, we are interested in establishing whether $I\land
\ltl{G}\,( Tr)\land\psi$ is satisfiable modulo the theory constraining the
database, i.e., $\mathrm{EUF}\cup \mathrm{LRA}$. As an example, 
imagine a job hiring business process modeling the evaluation of the
applications received for a job position: the transition system is intuitively
represented, via the BPMN standard
language,\footnote{\url{https://www.omg.org/spec/BPMN/}} in
Figure~\ref{fig:bpmn}. The following property states that whenever a job
position is opened, eventually it is closed.
\begin{align}
  \ltl{
    G(x_{\mathit{pos}}\neq\mathsf{null}_{\mathit{Job}}\rightarrow
        F(x_{\mathit{status}}=\mathsf{PosClosed})) 
          } 
\end{align}

\paragraph{Data-aware planning}

\LTLfMT satisfiability over LRA provides a straightforward approach to
\emph{numeric planning problems}, which are similarly
undecidable~\cite{Helmert02} but nevertheless very relevant in practice. Even
more interestingly, we can model \emph{data-aware} planning problems, \ie
\emph{planning problems}~\cite{GhallabNT04} where precoditions of actions can
query a relational database. As shown by Cialdea Mayer
\etal~[\citeyear{CialdeaMayerLOP07}], \LTL can encode classical STRIPS-like
planning problems. Similarly, data-aware planning problems can be encoded with
\LTLfMT formulas. The encoding of the planning domain results into a data-aware
transition system similar to the one above, with the property to verify
corresponding to the reachability of the goal.

\section{One-pass tree-shaped tableau for LTLf\textsuperscript{MT}}
\label{sec:tableau}

Here we provide a semi-decision procedure for \LTLfMT satisfiability based on an
adaptation of the one-pass tree-shaped tableau for \LTL by
Reynolds~[\citeyear{Reynolds16a}]. In contrast to classic graph-shaped
tableaux, Reynolds' produces a purely tree-shaped structure where a single pass
is sufficient to either accept or reject a given branch. Reynolds' tableau
proved to be amenable to efficient implementation~\cite{BertelloGMR16} and
parallelization~\cite{McCabeR17}, and to direct extension to real-time
logics~\cite{GeattiGMR21}. A state-of-the-art SAT encoding of Reynolds' tableau
has been implemented in the \BLACK satisfiability
solver~\cite{DBLP:conf/tableaux/GeattiGM19,GeattiGMV21}. Here, we adapt the
tableau system to \LTLfMT, and provide an SMT encoding that we implement in the
\BLACK satisfiability tool.

The \emph{closure} of a formula $\phi$, denoted $\closure(\phi)$, is the set of
all the subformulas of $\phi$ with the addition of $\ltl{X(\psi_1 U \psi_2)}$
for any formula $\ltl{\psi_1 U \psi_2}\in\closure(\phi)$, and $\ltl{wX(\psi_1 R
\psi_2)}$ for any formula $\ltl{\psi_1 R \psi_2}\in\closure(\phi)$.

A tableau for an \LTLfMT formula $\phi$ is a rooted tree where each node $u$ is
labelled by a set of formulas $\Gamma(u)\subseteq\closure(\phi)$. The root $u_0$
is labelled by $\Gamma(u_0)=\set{\phi}$. The tree is built starting from the
root, until all branches have been either accepted or rejected. If at least one
branch is accepted, $\phi$ is satisfiable. The tree is built by applying a set
of \emph{expansion rules} to the formulas in each node's label. Such rules are
listed in \cref{tab:expansion-rules}. When a rule is applied to a formula
$\psi\in\Gamma(u)$, two children $u'$ and $u''$ are added to $u$, where
$\Gamma(u')=\Gamma(u)\setminus\set{\psi}\cup\Gamma_1(u)$ and
$\Gamma(u'')=\Gamma(u)\setminus\set{\psi}\cup\Gamma_2(u)$, unless $\Gamma_2(u)$
is empty, in which case only a single child is added.

\begin{table}
  \begin{tabular}{>{\collectcell\trule}r<{\endcollectcell} fl fl fl} \toprule
    \multicolumn{1}{r}{\textsf{Rule}} &
      \psi\in\tlabel    & \Gamma_1(\psi) & \Gamma_2(\psi)                    \\
    \midrule
    Disjunction  & \alpha\lor\beta  & \set{\alpha}       & \set{\beta}       \\
    Conjunction  & \alpha\land\beta & \set{\alpha,\beta} &                   \\
    Until        & \alpha U\beta    & \set{\beta}        &
                                      \set{\alpha,X(\alpha U\beta)}        \\
    Release      & \alpha R\beta    & \set{\alpha,\beta} &
                                      \set{\beta,wX(\alpha R\beta)} \\
    \bottomrule
  \end{tabular}
  \caption{Expansion rules for the \LTLfMT tableau}
  \label{tab:expansion-rules}
\end{table}

When no expansion rules are applicable to a node, it is called a \emph{poised
node}. A poised node contains only \emph{elementary} formulas, \ie only
first-order, \emph{tomorrow} ($\ltl{X\phi}$) or \emph{weak tomorrow}
($\ltl{wX\phi}$) formulas. A poised node represents a single state in
a tentative model for the formula. So let $\branch=\seq{u_0,\ldots,u_{n-1}}$ be
a branch of the tableau with $u_{n-1}$ a poised node. A temporal step can be
made through the
\trule{Step} rule.
\begin{description}
  \item[\normalfont\trule{Step}] A new child $u_n$ of $u_{n-1}$ is created such 
  that:
    \begin{equation}
      \Gamma(u_n)=\set{\psi\mid\text{%
        $\ltl{X\psi}\in \Gamma(u_{n-1})$ or 
        $\ltl{wX\psi}\in \Gamma(u_{n-1})$}}
    \end{equation} 
\end{description}

By alternating the \trule{Step} rule with the expansion rules we build each
branch of the tree. Then, we need a way to stop the construction when a branch
is ready to be accepted or rejected. Two \emph{termination rules} are provided
for this purpose. The \trule{Contradiction} rule \emph{rejects} a branch when it
presents some contradiction. In order to define this rule, we need some
definitions. Given a branch $\branch=\seq{u_0,\ldots,u_{n-1}}$, with $u_{n-1}$ a
poised node, we define a first-order formula $\Omega(\branch)$ that summarizes
the first-order formulas made true by the branch $\branch$ over the different
time points. This formula is defined over a new alphabet $\Sigma' = \mathcal{P}'
\cup \mathcal{C} \cup \mathcal{F} \cup \mathcal{V}'\cup\mathcal W$ where
$\mathcal{P}'=\mathcal{P}\cup\set{\ell^i | i\in\N}$, and $\mathcal
V'=\set{x^i|x\in\mathcal{V}, i\in\N}$. Here, $\ell^i$ are fresh 0-ary EUF
predicates, and $\mathcal{V}'$ is a set of \emph{stepped} versions of the
variables in $\mathcal{V}$. Given a term $t$, the \emph{stepped} version of $t$
at time $i$ is the term $t^i$ defined as follows:
\begin{enumerate}
  \item $c^i=c$
  \item $x^i = x$ if $x\in\mathcal W$;
  \item $x^i = x^i$ if $x\in\mathcal V$;
  \item $(\nextvar x)^i=(\wnextvar x)^i=x^{i+1}$
  \item $(f(t_1,\ldots,t_k))^i=f((t_1)^i,\ldots,(t_k)^i)$
\end{enumerate}

By extension, given a first-order formula $\psi$, the formula $\psi^i$ is
obtained by replacing each term $t$ in $\psi$ with its stepped version $t^i$.
Given a first-order formula $\phi$, the formula $L_i(\phi)$ is obtained from
$\phi$ by replacing all \emph{strong} atoms $\alpha$ with $\ell^i \land \alpha$,
and all \emph{weak} atoms $\alpha$ with $\ell^i \implies \alpha$. Intuitively,
if $\ell^i$ is true, $i$ is not the last step of the model, hence $L_i(\alpha)$
encodes this aspect of the semantics of $\nextvar x$ and $\wnextvar x$ terms.

For the branch $\branch$, let $\bar\pi=\seq{\pi_0,\ldots,\pi_{m-1}}$ be the
sequence of \emph{poised nodes} of $\branch$. Let $F(\pi_i)$ be the set of
\emph{first-order} formulas of $\tlabel(\pi_i)$. Then we define the
$\Omega(\branch)$ formula as follows:
\begin{equation}
  \Omega(\branch)=\bigwedge_{i=0}^{m-1} 
    \smashoperator[r]{\bigwedge_{\psi\in F(\pi_i)}} (L_i(\psi))^i \land 
    \bigwedge_{i=0}^{m-2}\ell^i
\end{equation}

Intuitively, the purpose of $\Omega(\branch)$, which is a purely first-order
formula over $\mathcal{T}\cup\text{EUF}$, is that of describing the whole
tentative model represented by $\branch$ in a single formula.  Note that, in
$\Omega(\branch)$, the value of $\ell^{m-1}$ is left unconstrained. Hence, if
$\Omega(\branch)$ is unsatisfiable, a contradiction exists independently from
the choice of closing or extending the branch. Now we can define the rule.
\begin{description}
  \item[\normalfont\trule{Contradiction}] The branch $\branch$ is rejected if 
  $\Omega(\branch)$ is unsatisfiable modulo $\mathcal T\cup \text{EUF}$.
\end{description}

Note that the \trule{Contradiction} rule can be easily checked by giving
$\Omega(\branch)$ to an SMT solver over $\mathcal T\cup \text{EUF}$.

The \trule{Empty} rule, instead, \emph{accepts} suitable branches when 
there is no reason to further extend the model.
\begin{description}
  \item[\normalfont\trule{Empty}] If $\Gamma(\pi_{m-1})$ does \emph{not} contain
  \emph{tomorrow} formulas and $\Omega(\branch)\land\neg\ell^{m-1}$ is
  \emph{satisfiable}, the branch is accepted.
\end{description}

We can prove soundness and completeness of the system.
\begin{restatable}[Soundness and completeness]{theorem}{soundcomp}
  Given a \LTLfMT formula $\phi$, the tableau for $\phi$ has an accepted branch
  if and only if $\phi$ is satisfiable.
\end{restatable}

If $\mathcal T\cup \text{EUF}$ is decidable, the  breadth-first construction of the tableau
tree provides a semi-decision procedure for \LTLfMT satisfiability, as it always
finds at least an accepted branch when given a satisfiable formula. 

Despite the generality of our definitions, it is clear that our approach to
\LTLfMT satisfiability is applicable over  decidable theories 
supported by the
underlying SMT solvers. This is the case for most quantifier-free fragments of
the supported theories and their combinations, and, in some cases (such as LRA
and LIA), also in the quantified case.

The procedure may also terminate on some unsatisfiable formulas, such as those
where the unsatisfiability comes from theory contradictions, such as $\ltl{x=3 &
G(\exists y (x = y + y))}$. 
It cannot terminate instead on formulas where the unsatisfiability comes from
unsatisfiable temporal requests. An example of such formulas is $\ltl{G(x>3) &
F(x < 2)}$, whose tableau contains an infinite branch which tries to fulfill the
$\ltl{F(x < 2)}$ eventuality at each step. In Reynolds' tableau for
\LTL~\cite{Reynolds16a}, closing such branches is the purpose of the
\trule{Prune} rule, which however cannot be applied in \LTLfMT (otherwise we
would get a decision procedure for an undecidable problem). It is worth noticing
that the \trule{Prune} rule can be made to work if we do not allow $\nextvar x$
and $\wnextvar x$ terms. That is, a decision procedure is possible for the
fragment of \LTLfMT without $\nextvar x$ nor $\wnextvar x$ terms, which, however,
as we have seen, is far less expressive.


\section{The SMT encoding}
\label{sec:encoding}


\begin{table*}[!t]
  \small\centering
  \setlength{\sparkspikewidth}{1.5pt}
  \begin{tabular}{l @{\hspace{0.5em}}l @{\hspace{0.5em}}l @{\hspace{0.5em}}l}
    \toprule
    Theory & Formula & Result & Behavior \\\midrule
    LIA    & $\ltl{x = 0 & G(\wnextvar x = x + 1) & F(x = N)$} & SAT &
              \begin{sparkline}{8}
                \sparkspike 0.0 0.1
                \sparkspike 0.1 0.1
                \sparkspike 0.2 0.1
                \sparkspike 0.3 0.1
                \sparkspike 0.4 0.1
                \sparkspike 0.5 0.16
                \sparkspike 0.6 0.25
                \sparkspike 0.7 0.4
                \sparkspike 0.8 0.65
                \sparkspike 0.9 0.89
                \sparkspike 1.0 1.0
              \end{sparkline}
              {\tiny\sffamily\raisebox{1ex}{
              \begin{tabular}{@{}l@{}}
                4min50s\\
                N=2590
              \end{tabular}
              }
              }\\
           & $\ltl{
             x_0>0 & \bigwedge_{i=0}^{N-1}X^i(\nextvar x_{i+1} > x_i) &
             G(\bigwedge_{i=0}^N \wnextvar x_i = x_i) &
             G(\sum_{i=0}^{N-1} x_i = \frac{N(N-1)}{2}-1)
           }$ & UNSAT &
           \begin{sparkline}{8}
              \sparkspike 0.0 0.1
              \sparkspike 0.1 0.1 
              \sparkspike 0.2 0.1
              \sparkspike 0.3 0.1 
              \sparkspike 0.4 0.1 
              \sparkspike 0.5 0.1
              \sparkspike 0.6 0.1
              \sparkspike 0.7 0.1
              \sparkspike 0.8 0.1
              \sparkspike 0.9 0.4
              \sparkspike 1.0 1.0 
            \end{sparkline}
            {\tiny\sffamily\raisebox{1ex}{
              \begin{tabular}{@{}l@{}}
                4min30s\\
                N=57
              \end{tabular}
            }}
           \\\midrule
    LRA    & $\ltl{
              c = 1 & G(\wnextvar c = 10c) & X^N(x = c & G(\wnextvar x = \frac{x}{10}) & F(x = 1))
            }$ & SAT &
            \begin{sparkline}{8}
              \sparkspike 0.0 0.1
              \sparkspike 0.1 0.1
              \sparkspike 0.2 0.1
              \sparkspike 0.3 0.1
              \sparkspike 0.4 0.1
              \sparkspike 0.5 0.1
              \sparkspike 0.6 0.1
              \sparkspike 0.7 0.18
              \sparkspike 0.8 0.31
              \sparkspike 0.9 0.65
              \sparkspike 1.0 1.0
            \end{sparkline}
            {\tiny\sffamily\raisebox{1ex}{
              \begin{tabular}{@{}l@{}}
                4min43s\\
                N=234
              \end{tabular}
            }}
            \\[0.5ex]
           & $\begin{array}{@{}r@{}l@{}}
                c = 1 \land{} & 
                \ltl{G(\wnextvar c = 10c) & e = 1 & x = 0} \\
                {}\land{} & \ltl{
                  X^N\big(g = c & 
                  G(\wnextvar e = \frac{e}{2} & \wnextvar x = x + e & 0 \le 
                  x < 2) &
                  F(x > 2-\frac{1}{g})\big)}
              \end{array}$ & SAT &
            \begin{sparkline}{8}
              \sparkspike 0.0 0.1
              \sparkspike 0.1 0.1
              \sparkspike 0.2 0.1
              \sparkspike 0.3 0.1
              \sparkspike 0.4 0.1
              \sparkspike 0.5 0.1
              \sparkspike 0.6 0.1
              \sparkspike 0.7 0.1
              \sparkspike 0.8 0.17
              \sparkspike 0.9 0.56
              \sparkspike 1.0 1.0
            \end{sparkline}
            {\tiny\sffamily\raisebox{1ex}{
              \begin{tabular}{@{}l@{}}
                4min46s\\
                N=113
              \end{tabular}
            }}
            \\\midrule
    EUF+LIA & $ n = 0 \land c \ge 0 \land 
                \ltl{G}\left(\begin{array}{@{}l@{}}
                  \wnextvar c = c \land {}\\
                  \wnextvar n = n+1
                \end{array}\right) \land
                \ltl{G}\left(\begin{aligned}
                  n>1 &\to f(n) = 2f(n-1)+c \land {}\\
                  n=1 &\to f(n) = c
                \end{aligned}\right) \land \ltl{X^N(wX\bot)}
              $ & SAT &
                \begin{sparkline}{8}
                  \sparkspike 0.0 0.1
                  \sparkspike 0.1 0.1
                  \sparkspike 0.2 0.1
                  \sparkspike 0.3 0.1
                  \sparkspike 0.4 0.1
                  \sparkspike 0.5 0.1
                  \sparkspike 0.6 0.13
                  \sparkspike 0.7 0.29
                  \sparkspike 0.8 0.41
                  \sparkspike 0.9 0.59
                  \sparkspike 1.0 1.0
                \end{sparkline}
                {\tiny\sffamily\raisebox{1ex}{
                  \begin{tabular}{@{}l@{}}
                    4min28s\\
                    N=387
                  \end{tabular}
                }}\\\bottomrule
  \end{tabular}
  \caption{%
    Test formulas used in the experimental evaluation. $N$ is the scalable 
    parameter.
  }
  \label{tab:benchmarks}
\end{table*}

The SAT encoding of the original tableau for \LTL by Reynolds has been
implemented in the \BLACK satisfiability checking tool as described by Geatti
\etal~[\citeyear{DBLP:conf/tableaux/GeattiGM19,GeattiGMV21}]. Here, we extend it
naturally, namely with an SMT encoding of the tableau system described above.

We encode the tableau for a formula $\phi$ as a pair of SMT formulas that
represent the branches of the tree up to a depth $k$, for increasing values of
$k$, as shown in \cref{algo:black}.
The formula $\unr{\phi}_k$ used at Line 4 is called the \emph{k-unraveling} of
$\phi$, and represents the branches of the tree with at most $k+1$ poised nodes.
To define it, we encode the expansion rules of
\cref{tab:expansion-rules}.\fitpar

\begin{algorithm}[t]
  \caption{BLACK's main procedure for \LTLfMT}
  \label{algo:black}
  \begin{algorithmic}[1]
    \Procedure{\texttt{BLACK}}{$\phi$}
      \State $k \gets 0$
        \While{$\mathit{True}$}
          \If{$\unr{\phi}_k$ is \UNSAT}
            \State \textbf{return} $\phi$ is \UNSAT
          \EndIf
          
          \If{$\encod{\phi}_k$ is \SAT}
            \State \textbf{return} $\phi$ is \SAT
          \EndIf
          
          \State $k \gets k+1$
        \EndWhile
    \EndProcedure
  \end{algorithmic}
\end{algorithm}

\begin{definition}[Stepped Normal Form]
  Given an \LTLfMT formula $\phi$ and an $i\ge0$, its $i$-th \emph{stepped
  normal form}, denoted by $\snf_i(\phi)$, is defined as follows:
  \begin{align}
    \snf_i(\lambda) &= L_i(\lambda)
      \tag*{where $\lambda$ is a first-order formula} \\
    \snf_i(\otimes\,\phi_1) &= \otimes\,\phi_1 
      \tag*{where $\otimes \in \set{\ltl{X}, \ltl{wX}}$}      \\
    \snf_i(\phi_1 \otimes \phi_2) &= \snf_i(\phi_1) \otimes \snf_i(\phi_2) 
      \tag*{where $\otimes \in \set{\land, \lor}$}       \\
    \snf_i(\ltl{\phi_1 U \phi_2}) &= \snf_i(\phi_2) \lor (\snf_i(\phi_1) \land \ltl{X(\phi_1 U \phi_2)})  \\
    \snf_i(\ltl{\phi_1 R \phi_2}) &= \snf_i(\phi_2) \land (\snf_i(\phi_1) \lor \ltl{wX(\phi_1 R \phi_2)})
  \end{align}
\end{definition}
For a generic formula $\psi$ and for $i>0$, considering its stepped version
$\psi^i$, we denote as $\psi^i_G$ the result of replacing each \emph{tomorrow}
or \emph{weak tomorrow} formula $\tau$ with a fresh 0-ary EUF predicate
$\tau^i_G$.

The $k$-unraveling of $\phi$ is recursively defined as follows:
\begin{align}
  \unr{\phi}_0 = {}&\snf_0(\phi)_G^0 \\
  \unr{\phi}_{k+1} = {}&\unr{\phi}_k \land \ell^k \land 
    \smashoperator{\bigwedge_{\ltl{X \alpha} \in \XR}} 
    \Big( (\ltl{X \alpha})_G^k \leftrightarrow \snf_{k+1}(\alpha)_G^{k+1}\Big)\\
    &\phantom{\unr{\phi}^k\land \ell^k} \land 
    \smashoperator{\bigwedge_{\ltl{wX \alpha} \in \wXR}}
    \Big( (\ltl{wX \alpha})_G^k \leftrightarrow \snf_{k+1}(\alpha)_G^{k+1}\Big)
\end{align}
where $\XR$ and $\wXR$ are the sets of all the \emph{tomorrow} and \emph{weak
tomorrow} formulas, respectively, in $\closure(\phi)$. 

If $\unr{\phi}_k$ is unsatisfiable, it means all the branches of the tableau for
$\phi$ are of at most $k+1$ poised nodes and are all rejected by the
\trule{Contradiction} rule.

The $\encod{\phi}_k$ formula encodes the \trule{Empty} rule:
\begin{equation}
  \encod{\phi}_k \equiv \unr{\phi}_k \land 
    \bigwedge_{\psi \in \XR} \neg \psi_G^k \land \neg\ell^k
\end{equation}
If $\encod{\phi}_k$ is satisfiable, the tableau for $\phi$ has at least a branch
accepted by the \trule{Empty} rule. 

By relating the tableau branches with the models for $\unr{\phi}_k$ and
$\encod{\phi}_k$, one can easily prove the following, taking inspiration from \cite{GeattiGMV21}.
\begin{restatable}{theorem}{encodingthm}
  \Cref{algo:black} answers \SAT if and only if the tableau for $\phi$ has an
  accepted branch.
\end{restatable}

\section{Experiments}
\label{sec:experiments}

We implemented the above encoding in the \BLACK satisfiability checking
tool.\footnotemark\ The current implementation supports the full \LTLfMT syntax
limited to mono-sorted theories, but in addition to that, it supports past
operators, which are handled in the tableau similarly to \cite{GeattiGMR21} and
are encoded similarly to~\cite{GeattiGMV21}.

\footnotetext{%
  The source code and all the tests and benchmarks are available in \BLACK's
  GitHub repository, merged into version \textsf{0.7} of the tool.%
}

We tested \BLACK on a number of crafted benchmarks intended to stress the
underlying algorithm in different ways. We have a total of 990 formulas
generated from 5 scalable patterns each parameterized on a value $N$. The
scalable schemata are shown in \cref{tab:benchmarks}. The table shows the theory
over which the formulas are intended to be tested, the expected result of the
satisfiability checking, and a qualitative plot of running times from $N=1$ up
to the maximum value of $N$ which runs under the 5 minutes timeout. The tests
were run on commodity hardware, setting \BLACK to use the Z3
backend~\cite{MouraB08}.

The first LIA schema represents a set of simple counters going from 0 to $N$.
The tableau for these formulas has an accepted branch of exponential size, but
\BLACK's encoding is not paying a price to that. In contrast, the explicit
Boolean encoding of a counter found in \BLACK's propositional \LTL benchmark
set~\cite{DBLP:conf/tableaux/GeattiGM19} takes orders of magnitude longer to
solve in the harder cases. This unsurprising fact shows the advantage of the
explicit arithmetic reasoning done by the SMT solver over bit blasting
techniques. The second LIA schema is an example of \emph{unsatisfiable} formulas
where \BLACK nevertheless halts, and stresses the behavior of the solver when
nontrivial arithmetic is involved. 

The LRA schemata represent behaviors that are only possible in a dense domain.
Both build an arbitrary big integer value used to define an arbitrary small
fraction to use as the target of the computation. The combined EUF+LIA schema
recursively defines a function (inspired by the computational complexity proofs
of common sorting algorithms) and forces the model to compute it for the given
number of steps.

We can see that these formulas stress the solver but are solved in a manageable
amount of time. Globally, this gives evidence of the applicability of our
approach to data-aware reasoning.


\section{Conclusions}
\label{sec:conclusions}

This paper presented \LTLfMT, an extension of \LTLf where propositional symbols
are replaced by first-order formulas over arbitrary theories, \ala SMT. We
discussed how the logic can be useful to approach the model-checking problem for
data-aware systems and a form of data-aware planning problems. Although \LTLfMT
is easily seen to be undecidable, we provided a semi-decision procedure in the
form of a sound and complete tableau system \ala Reynolds, that we encode into
an incremental SMT-based algorithm. We then implemented our approach in \BLACK,
a state-of-the-art \LTL/\LTLf satisfiability checker, proving the applicability
of the approach with a number of crafted benchmark tests. 

Investigating and developing the use cases discussed in \cref{sec:use-cases}
seems the most natural evolution of this work.

\clearpage
\section*{Acknowledgements} This work was supported by the projects VERBA, MENS, TOTA and STAGE by the Free University of Bozen-Bolzano. 
\bibliographystyle{named}
\bibliography{ijcai22}

\ifappendix
  \appendix
  \clearpage

\section{Proofs}

\subsection*{Expressiveness of LTLf\textsuperscript{MT}}

\booleanabstraction*

\begin{proof}
  The right-to-left direction is straightforward, since \LTLf is subsumed by 
  \LTLfMT.

  For the left-to-right direction, consider an \LTLfMT formula $\phi$ without
  $\nextvar x$ nor $\wnextvar x$ terms. Given the signature $\Sigma
  = \mathcal{S} \cup \mathcal{P} \cup \mathcal{C} \cup \mathcal{F} \cup
  \mathcal{V}\cup\mathcal{W}$, we call $\mathcal{AP}_{\Sigma}$ the set of
  atomic propositions such that $P \in \mathcal{AP}_{\Sigma}$ iff $P \in
  \mathcal{P}$ is a \emph{unary} predicate symbol (\emph{w.l.o.g.} we can
  assume that $\mathcal{AP}_{\Sigma} \neq \emptyset$).  Given a $\theory$-state
  $s = \pair{M,\mu}$, we define $\mathbb{B}(s)$ (the \emph{Boolean abstraction}
  of $s$) as \emph{the} subset of $\mathcal{AP}_{\Sigma}$ such that:
  \begin{align}
    P \in \mathbb{B}(s) \Iff s \models \exists x \suchdot P(x)
  \end{align}
  for all $P \in \mathcal{AP}_{\Sigma}$. We extend the notion of \emph{Boolean
  abstraction} to words and languages as follows.  Given a finite word $\sigma
  = \seq{s_1,\dots,s_n}$ over the theory $\theory$, we define
  $\mathbb{B}(\sigma)$ (the \emph{Boolean abstraction} of $\sigma$) as the word
  $\seq{\mathbb{B}(s_1),\dots,\mathbb{B}(s_n)}$ over
  $2^{\mathcal{AP}_{\Sigma}}$. Similarly, given a language $\lang$
  modulo the theory $\theory$ (recall \cref{sec:ltlfmt}), we define
  $\mathbb{B}(\lang) = \set{\mathbb{B}(\sigma) \in
  (2^{\mathcal{AP}_{\Sigma}})^* \suchthat \sigma \in \lang}$.

  We will build an \LTLf formula $\phi_{\mathbb{B}}$ over the set of atomic
  propositions $\mathcal{AP}_{\Sigma}$ such that $\lang(\phi_{\mathbb{B}})
  = \mathbb{B}(\lang(\phi))$. Consider any first-order subformula $\psi$ of
  $\phi$. We define $\mathbb{B}(\psi)$ (the Boolean abstraction of $\psi$) as
  the set of the Boolean abstractions of all the $\theory$-states that are
  models of $\psi$ (\ie $\mathbb{B}(\psi) = \set{\mathbb{B}(s)}_{s \models
  \psi}$)\footnotemark.
  \footnotetext{In the general case, there exist infinitely many
  $\theory$-states $s$ such that $s \models \psi$, since $M$ may be a domain of
  infinite cardinality. This makes the proof not constructive.}
  Since any Boolean abstraction is a subset of $\mathcal{AP}_{\Sigma}$ and
  since $\mathcal{AP}_{\Sigma}$ is finite, the cardinality of
  $\mathbb{B}(\psi)$ is finite. This allows us to define the following Boolean
  formula over the variables $\mathcal{AP}_{\Sigma}$:
  \begin{align}
    \psi_{\mathbb{B}} \coloneqq \bigvee_{S \in \mathbb{B}(\psi)}
    \bigwedge_{\psi^\prime \in S} \psi^\prime
  \end{align}
  By definition of $\psi_{\mathbb{B}}$, for any $\theory$-state $s$, it holds
  that $\mathbb{B}(s) \models \psi_{\mathbb{B}}$ iff $s \models \psi$.  Let
  $\phi_{\mathbb{B}}$ be the \LTLf formula obtained by replacing all the
  maximal first-order subformulas $\psi$ of $\phi$ with $\psi_\mathbb{B}$. By
  induction, it can be easily proved that $\lang(\phi_{\mathbb{B}})
  = \mathbb{B}(\lang(\phi))$. This concludes the proof.
\end{proof}

\subsection*{Tableau for LTLf\textsuperscript{MT}}

Soundness and completeness are shown in the following two lemmas. Let us define 
some notation. For a node $u$ in the tableau for an \LTLfMT formula, we define \begin{equation}
  \Gamma^*(u)=\set*{
    \psi\in\Gamma(u') \middle| \begin{aligned}
      &\text{$u'$ is $u$ or an ancestor of $u$ descendant}\\[-1ex]
      &\text{of the poised node closest to $u$}\\[-1ex]
      &\text{or any ancestor of $u$ if there are no}\\[-1ex]
      &\text{poised nodes ancestors of $u$}
    \end{aligned}
  }
\end{equation}
For  signatures $\Sigma = \mathcal{S} \cup \mathcal{P} \cup
\mathcal{C} \cup \mathcal{F} \cup \mathcal{V}\cup\mathcal{W}$ and $\Sigma' = \mathcal{P}'
\cup \mathcal{C} \cup \mathcal{F} \cup \mathcal{V}'\cup\mathcal W$ as above, we call $\theory'$ the class of
$\Sigma'$-structures whose $\Sigma$-reduct is in $\theory$, and \emph{$\theory'$-state} a pair $s = \pair{M,\mu}$ made of a $\Sigma'$-structure
$M \in \theory'$ and a variable evaluation function $\mu : \mathcal{V}' \to
\dom(M)$.
For a purely
first-order term $t$, a $\theory'$-state $s$ and an environment $\xi$, we use the
shortcut notation $\eval{t}_{s,\xi}=\eval{t}^0_{\sigma,\xi}$ where $\sigma$ is
any word whose first state is $s$.

\begin{lemma}[Soundness]
  \label{thm:soundness}
  Given an \LTLfMT formula $\phi$, if the complete (i.e., the completely constructed) tableau  for $\phi$ contains 
  an accepted branch, then $\phi$ is satisfiable.
\end{lemma}
\begin{proof}
  We show how to extract a model for $\phi$ from an accepted branch $\branch$.
  Let $\bar\pi=\seq{\pi_0,\ldots,\pi_{m-1}}$ be the sequence of poised nodes of
  $\branch$. Note that since the branch is accepted, the \trule{Empty} rule
  fired, hence $\Omega(\branch)\land\neg\ell^{m-1}$ is satisfiable in $\mathcal
  T\cup\text{EUF}$. This means we have a $\Sigma'$-structure $M$ whose $\Sigma$-reduct is in $\mathcal{T}$
  and a variable evaluation function $\mu:\mathcal{V}'\to\dom(M)$ such that
  $s\models\Omega(\branch)\land\neg\ell^{m-1}$ where $s=(M,\mu)$. From $\mu$ we
  extract a model $\sigma$ for $\phi$ as follows. We define
  $\sigma=\seq{(M,\mu_0),\ldots, (M,\mu_{m-1})}$ such that $\mu_i(x)=\mu(x^i)$
  for all $0\le i < m$ and all $x\in\mathcal{V}$. Now we prove that $\sigma$ is
  indeed a model for $\phi$. To do that, we first show the relationship between
  terms in $\phi$ and terms in $\Omega(\phi)$, and in particular we prove that
  $\eval{t}_{\sigma,\xi}^i=\eval{t^i}_{s,\xi}$ for any environment $\xi$. This
  is done by induction on the structure of terms:
  \begin{enumerate}
    \item $\eval{c}_{\sigma,\xi}^i=c^M=\eval{c}_{s,\xi}$;
    \item for $x\in \mathcal V$, $\eval{x}_{\sigma,\xi}^i=\mu_i(x)$ and
    $\mu_i(x)=\mu(x^i)=\eval{x^i}_{s,\xi}$ by definition;
    \item for $x\in \mathcal W$, 
      $\eval{x}_{\sigma,\xi}^i=\xi(x)=\eval{x}_{s,\xi}=\eval{x^i}_{s,\xi}$ by 
      definition;
    \item
    $\eval{f(t_1,\ldots,t_k)}_{\sigma,\xi}^i=f^M(\eval{t_1}_{\sigma,\xi}^i,\ldots,\eval{t_k}_{\sigma,\xi}^i)$
    by definition and by the inductive hypothesis
    $\eval{t_j}_{\sigma,\xi}^i=\eval{t_j^i}_{s,\xi}$, hence:
    \begin{align}
      f^M(\eval{t_1}_{\sigma,\xi}^i,\ldots,\eval{t_k}_{\sigma,\xi}^i) &=
      f^M(\eval{t_1^i}_{s,\xi},\ldots,\eval{t_k^i}_{s,\xi})\\
      &=\eval{f(t_1^i,\ldots,t_k^i)}_{s,\xi}\\
      &= \eval{(f(t_1,\ldots,t_k))^i}_{s,\xi}
    \end{align}
    \item $\eval{\nextvar x}_{\sigma,\xi}^i=\eval{x}_{\sigma,\xi}^{i+1}$ by 
    definition, and
    by the inductive hypothesis we know that
    $\eval{x}_{\sigma,\xi}^{i+1}=\eval{x^{i+1}}_{s,\xi}=
      \eval{(\nextvar x)^i}_{s,\xi}$. The case for $\wnextvar x$ is similar.
  \end{enumerate}

  Now we go by induction by proving that if it holds that, for all $0\le i < m$:
  \begin{equation}
    s,\xi\models (L_i(\psi))^i \land \bigwedge_{j=0}^{m-2} \ell^j\land
      \neg\ell^{m-1}
  \end{equation}
  then $\sigma,\xi,i\models\psi$ for any first-order formula $\psi$ and any
  environment $\xi$. For the base case, we have $s,\xi\models
  (L_i(R(t_1,\ldots,t_k)))^i$. We distinguish the case of whether
  $R(t_1,\ldots,t_k)$ is a \emph{weak} atom, a \emph{strong} atom, or neither:
  \begin{enumerate}
    \item if it is a \emph{weak} atom, by definition of $L_i$ we have
      $s,\xi\models \ell^i\implies R(t_1^i,\ldots,t_k^i)$. If $M$ evaluates
      $\ell^i$ to false, then it must be $i=m-1$ because all the others $\ell^j$
      are set to true. Since the atom is weak, by semantics it trivially holds
      at the last state, hence $\sigma,\xi,i\models R(t_1,\ldots,t_k)$. Now
      suppose $M$ assigns $l^i$ to true. In this case, we know that $i<m-1$,
      because $\ell^{m-1}$ is set to false. Hence,
      $\eval{t_1^i}_{s,\xi},\ldots,\eval{t_k^i}_{s,\xi}$ are well defined. Since
      the implication holds and $\ell^i$ is true, by definition we have that
      $s,\xi\models R(t_1^i,\ldots,t_k^i)$, that is
      $(\eval{t_1^i}_{s,\xi},\ldots,\eval{t_k^i}_{s,\xi})\in R^M$. But
      $\eval{t_1^i}_{s,\xi}=\eval{t_1}_{\sigma,\xi}^i$ as we proved above, so
      $(\eval{t_1}_{\sigma,\xi}^i,\ldots,\eval{t_k}_{\sigma,\xi}^i)\in R^M$
      which by definition means $\sigma,\xi,i\models R(t_1,\ldots,t_k)$.
    \item if it is a \emph{strong} atom, then we have $s,\xi\models\ell^i \land
      R(t_1^i,\ldots,t_k^i)$. This means in particular that
      $s,\xi\models\ell^i$. Hence, we know that $i<m-1$, because $\ell^{m-1}$ is
      false, therefore $\eval{t_1^i},\ldots,\eval{t_k^i}$ are well-defined. Now,
      we have that $s,\xi\models R(t_1^i,\ldots,t_k^i)$, that is
      $(\eval{t_1^i}_{s,\xi},\ldots,\eval{t_k^i}_{s,\xi})\in R^M$. But
      $\eval{t_1^i}_{s,\xi}=\eval{t_1}_{\sigma,\xi}^i$ as we proved above, so
      $(\eval{t_1}_{\sigma,\xi}^i,\ldots,\eval{t_k}_{\sigma,\xi}^i)\in R^M$
      which by definition means $\sigma,\xi,i\models R(t_1,\ldots,t_k)$.
    \item if it is neither, $L_i$ is an identity, so we know $s,\xi\models
    R(t_1^i,\ldots,t_k^i)$, that is
    $(\eval{t_1^i}_{s,\xi},\ldots,\eval{t_k^i}_{s,\xi})\in R^M$. But
    $\eval{t_1^i}_{s,\xi}=\eval{t_1}_{\sigma,\xi}^i$ as we proved above, so
    $(\eval{t_1}_{\sigma,\xi}^i,\ldots,\eval{t_k}_{\sigma,\xi}^i)\in R^M$
    which by definition means $\sigma,\xi,i\models R(t_1,\ldots,t_k)$.
  \end{enumerate}
  For the inductive cases:
  \begin{enumerate}
    \item If $s,\xi\models (L_i(\exists x\lambda))^i\land \bigwedge_{j=0}^{m-2}
    \ell^j\land \neg\ell^{m-1}$, it means that there exists a value
    $v\in\dom(M)$ such that $s,\xi[x\leftarrow v]\models (L_i(\lambda))^i\land
    \bigwedge_{j=0}^{m-2} \ell^j\land \neg\ell^{m-1}$. Then, by induction
    hypothesis, we know that $\sigma,\xi[x\leftarrow v],i\models\lambda$, that
    by definition means that $\sigma,\xi,i\models\exists x\lambda$. The case for
    $\forall x\lambda$ is similar.
    \item the other inductive cases (negation, disjunction, conjunction) are 
      trivial.
  \end{enumerate}
  Now we can prove by induction that for any formula $\psi\in\Gamma^*(u_i)$ it
  holds that $\sigma,\xi,i\models\psi$. For the base case, for a first-order
  formula $\psi\in\Gamma^*(u_i)$ (hence in $\Gamma(u_i)$ as well), we know that
  $(L_i(\psi))^i$ is a conjunct of $\Omega(\branch)$ (\emph{w.l.o.g.} we can
  assume $\psi$ is not a conjunction or a disjunction because those are expanded
  by the tableau), hence $s,\xi\models (L_i(\psi))^i\land \bigwedge_{j=0}^{m-2}
  \ell^j\land \neg\ell^{m-1}$ (recall that the branch is accepted hence the
  \trule{Empty} rule fired) for any environment $\xi$. By what we proved above,
  this implies that $\sigma,\xi,i\models\psi$. The inductive cases are trivial.

  In particular this proves the thesis because $\phi\in\Gamma^*(u_0)$.
\end{proof}

Completeness is easier than the propositional case because we do not have the
\trule{Prune} rule.

\begin{lemma}[Completeness]
  Given an \LTLfMT formula $\phi$, if $\phi$ is satisfiable, then the complete
  tableau for $\phi$ contains an accepted branch.
\end{lemma}
\begin{proof}
  Let $\sigma$ be a model for $\phi$. We use $\sigma$ as a guide to descend the
  complete tableau for $\phi$, looking for an accepted branch. We start from the
  root $u_0$ and we find a branch $\branch=\seq{u_0,\ldots,u_{n-1}}$ by choosing
  at each step one of the children of the current node $u_i$ to become
  $u_{i+1}$. During the descent, we maintain a function $J:\N\to\N$ from
  positions in the current branch to positions in the model, with the invariant
  that $\sigma,J(i)\models\psi$ for all $\psi\in\Gamma(u_i)$. Initially, we
  define $J(0)=0$, and since $\Gamma(u_0)=\set{\phi}$, the invariant holds by
  definition. Then, at each step we choose a child of $u_i$ to become
  $u_{i+1}$ as follows:
  \begin{enumerate}
    \item If an expansion rule has been applied to $u_i$, then we have one or
    two children. If we have only one child, we pick that as $u_{i+1}$, and we
    define $J(i+1)=J(i)$. The invariant holds since by contruction of the
    expansion rules, we have that $\Gamma(u_i)\models\Gamma(u_{i+1})$, hence
    since we had that $\sigma,J(i)\models\Gamma(u_i)$, we have
    $\sigma,J(i+1)\models\Gamma(u_{i+1})$. If we have two children $u_i'$ and
    $u_i''$, we know by construction of the expansion rules that either
    $\sigma,J(i)\models\Gamma(u_i')$ or $\sigma,J(i)\models\Gamma(u_i'')$, hence
    we choose $u_{i+1}$ accordingly and we set $J(i+1)=J(i)$.
    \item If the \trule{Step} rule was applied to $u_i$, we have only one child
    that we pick as $u_{i+1}$, and we set $J(i+1)=J(i)+1$. Since
    $\sigma,J(i)\models\Gamma(u_i)$, it holds that
    $\sigma,J(i)+1\models\Gamma(u_{i+1})$ by the semantics of the
    \emph{tomorrow} and \emph{weak tomorrow} operators and by the definition of
    the \trule{Step} rule. Hence, $\sigma,J(i+1)\models\Gamma(u_{i+1})$, \ie the
    invariant holds.
  \end{enumerate}

  The descent proceeds until reaching a leaf or the end of the model $\sigma$.
  Now, let $\branch=\seq{u_0,\ldots,u_{n-1}}$ be the branch found by this
  descent and let $\bar\pi=\seq{\pi_0,\ldots,\pi_{m-1}}$ be the sequence of
  poised nodes of $\branch$. Note that $m=\abs{\sigma}$. Now, we can show that
  $\branch$ is indeed an accepted branch. To do that, we extract from $\sigma$ a
  model for $\Omega(\branch)\land\neg\ell^{m-1}$, thus showing that the
  \trule{Empty} rule has accepted the branch (and, consequently, that the
  \trule{Contradiction} rule has not fired). If
  $\sigma=\seq{(M,\mu_0),\ldots,(M,\mu_{m-1})}$, we define the $\mathcal
  T'$-state $s=(M',\mu)$ where $\mu(x^i)=\mu_i(x)$ and $M'$ agrees with $M$ on
  everything and moreover evaluates $\ell^j$ as true for $0\le j < m - 1$ and
  $\ell^{m-1}$ as false. Note that the invariant maintained during the
  descent translates to the fact that $\sigma,i\models\psi$ for all
  $\psi\in\Gamma(\pi_i)$. 
  
  Similarly to the proof of \cref{thm:soundness}, we can prove by induction that
  for all terms $t$, $\eval{t^i}_{s,\xi}=\eval{t}_{\sigma,\xi}^i$ for any
  environment $\xi$. Now, let $\psi$ be a first-order formula. We can prove by
  induction that, for all $0\le i < m$, if $\sigma,\xi,i\models\psi$ for some
  environment $\xi$, then $s,\xi\models (L_i(\psi))^i$.
  \begin{enumerate}
    \item If $\sigma,\xi,i\models R(t_1,\ldots,t_k)$, then
      $(\eval{t_1}_{\sigma,\xi}^i,\ldots,\eval{t_k}_{\sigma,\xi}^i)\in R^M$. By
      what we shown above this means that
      $(\eval{t_1^i}_{s,\xi},\ldots,\eval{t_k^i}_{s,\xi})\in R^M$, hence $(\eval{t_1^i}_{s,\xi},\ldots,\eval{t_k^i}_{s,\xi})\in R^{M'}$ as well, that by
      definition means that $s,\xi\models R(t_1^i,\ldots,t_k^i)$. This also
      implies that $s,\xi\models \ell^i \implies R(t_1^i,\ldots,t_k^i)$, hence,
      if $R(t_1^i,\ldots,t_k^i)$ is a \emph{weak} atom, we have $s,\xi\models
      (L_i(\psi))^i$. If instead $R(t_1^i,\ldots,t_k^i)$ is a \emph{strong} atom,
      we know that $\eval{t_1^i},\ldots,\eval{t_k^i}$ are well-defined, hence
      $i<m-1$, which means that $s,\xi\models \ell^i \land
      R(t_1^i,\ldots,t_k^i)$, since by construction $\ell^j$ holds for all $0\le
      j < m -1$, hence $s,\xi\models (L_i(\psi))^i$ in this case as well. If the
      atom is neither strong nor weak, we already have that $s,\xi\models
      (L_i(\psi))^i$.
    \item If $\sigma,\xi,i\models \exists x \psi$, it means that there exists a
      value $v\in\dom(M)$ such that $\sigma,\xi[x\leftarrow v],i\models\psi$.
      Then, by the induction hypothesis, we know that $s,\xi[x\leftarrow
      v]\models (L_i(\psi))^i$, which by definition means that $s,\xi\models
      L_i((\exists x\psi)^i)$. The case for $\forall x\psi$ is similar.
  \end{enumerate} 
  
  Since by the invariant of the descent we know that $\sigma,i\models\psi$ for
  all $\psi\in\tlabel(\pi_i)$, by what we showed above this means $s,\xi\models
  (L_i(\psi))^i$ for all $\psi\in F(\pi_i)$, then by construction we also
  have that $s,\xi\models\ell^j$ for all $0\le j < m-1$ and
  $s,\xi\models\neg\ell^{m-1}$, hence by recalling the definition of
  $\Omega(\branch)$, we showed that
  $s,\xi\models\Omega(\branch)\land\neg\ell^{m-1}$. As a last consideration,
  note that if we had any \emph{tomorrow} formula
  $\ltl{X\psi}\in\Gamma(\pi_{m-1})$, that would mean $\sigma,{m-1}\models
  \ltl{X\psi}$, but then $m-1$ could not be the last position of the model.
  Hence, the \trule{Empty} rule fired to accept the branch.
\end{proof}

Together, the above lemmas prove the following.
\soundcomp*

\subsection*{Encoding}

Here we prove the correctness of the SMT encoding of our tableau.

First, we define the set $\mathcal{P}''=\mathcal{P}\cup\set{\ell^i | i\in\N}\cup\set{(\ltl{X\alpha})^i_G |\text{$\alpha$  is an \LTLfMT formula}}$.

\begin{definition}[Atom]
  \label{def:pre-model-atom}
  Given a formula $\phi$, an \emph{atom} for $\phi$ is a set
  $\Delta$ of \emph{purely first-order} formulas 
  over $P''$ and
  $V'$ defined as follows:
  \begin{enumerate}
    \item either $p(t_1,\ldots,t_n)\in\Delta$ or 
      $\neg p(t_1,\ldots,t_n)\in\Delta$;
    \item $\psi_1\lor\psi_2\in\Delta$ iff $\psi_1\in\Delta$ or 
      $\psi_2\in\Delta$;
    \item $\psi_1\land\psi_2\in\Delta$ iff $\psi_1\in\Delta$ and 
      $\psi_2\in\Delta$;
    \item $\exists x \phi(x)\in\Delta$ iff there is a term $t$ such that  
      $\phi(t)\in\Delta$;
    \item $\forall x \phi(x)\in\Delta$ iff for all terms $t$ we have
      $\phi(t)\in\Delta$;
  \end{enumerate}
\end{definition}

Given a sequence $\bar\Delta=\seq{\Delta_0,\ldots,\Delta_m}$ of atoms, we define
$\Omega(\bar\Delta)$ as the following set of formulas:
\begin{equation}
  \Omega(\bar\Delta)=\bigcup_{i=0}^m \Delta_i \cup 
    \set{\ell^0,\ldots,\ell^{m-1}}
\end{equation}

\begin{definition}[Prefix pre-model]
  \label{def:prefix-pre-model}
  Given a \LTLfMT formula $\phi$, a \emph{prefix pre-model} for $\phi$ is a
  sequence $\bar\Delta=\seq{\Delta_0,\ldots,\Delta_m}$ of atoms such that:
  \begin{enumerate}
    \item $\snf_0(\phi)^0_G\in\Delta_0$
    \item for $i< m$, $(\ltl{X\alpha})^i_G\in\Delta_i$ iff
      $\snf_{i+1}(\alpha)^{i+1}_G\in\Delta_{i+1}$;
    \item for $i< m$, $(\ltl{wX\alpha})^i_G\in\Delta_i$ iff
    $\snf_{i+1}(\alpha)^{i+1}_G\in\Delta_{i+1}$;
    \item $\Omega(\bar\Delta)$ is satisfiable.
  \end{enumerate}
\end{definition}

\begin{lemma}
  \label{lemma:pre-model-existence}
  Let $\branch$ be a non-rejected prefix of any branch of the complete tableau
  for $\phi$, $\bar\pi=\seq{\pi_0,\ldots,\pi_m}$ be the sequence of its poised
  nodes, and let $s\models\Omega(\branch)$. There exists at least a prefix
  pre-model $\bar\Delta=\seq{\Delta_0,\ldots,\Delta_m}$ for $\phi$ such
  that for each formula $\psi\in\closure(\phi)$ and for each $i\le m$, we have
  $\snf_i(\psi)^i_G\in\Delta_i$ iff $s\models\snf_i(\psi)^i_G$.
\end{lemma}
\begin{proof}
  We now define $\bar\Delta=\seq{\Delta_0,\ldots,\Delta_m}$ to be a prefix
  pre-model for $\phi$ such that $\ell^i\in\Delta_i$ for all $0\le i < m$ and
  for all atoms $p(t_1^i,\ldots,t_n^i)$ we have
  $p(t_1^i,\ldots,t_n^i)\in\Delta_i$ if and only if $s\models
  p(t_1^i,\ldots,t_n^i)$. We now have to show that $\bar\Delta$ is the prefix
  pre-model we are looking for. First, note that such a prefix pre-model exists
  (\ie $\Omega(\bar\Delta)$ is satisfiable). To see this, consider that $s$ is
  also a model of $\Omega(\bar\Delta)$, and \emph{w.l.o.g.} we can suppose it to
  also evaluate $s\models(\ltl{X\alpha})^i_G$ if and only if
  $(\ltl{X\alpha})^i_G\in\Delta_i$ ($(\ltl{X\alpha})^i_G$ propositions are
  unconstrained by $\Omega(\branch)$) and similarly for $(\ltl{wX\alpha})^i_G$.
  
  Now, we can show that for all first-order formulas $\lambda$, we have
  $L_i(\lambda)^i\in\Delta_i$ if and only if $s\models L_i(\lambda)^i$. We go 
  by induction on $\lambda$:
  \begin{enumerate}
    \item for the base case of an atom, if
    $L_i(p(t_1,\ldots,t_n))^i\in\Delta_i$, we distinguish the cases of whether
    the atom is \emph{strong}, \emph{weak}, or neither:
    \begin{enumerate}
      \item if the atom is strong, then we have $\ell^i\land
      p(t_1^i,\ldots,t_n^i)\in\Delta_i$, hence we have $\ell^i\in\Delta_i$ and
      $p(t_1^i,\ldots,t_n^i)\in\Delta_i$. By definition we then have that
      $s\models\ell^i$ and $s\models p(t_1^i,\ldots,t_n^i)$, hence
      $s\models\ell^i\land p(t_1^i,\ldots,t_n^i)$ which means $s\models
      L_i(p(t_1,\ldots,t_n))^i$. The \viceversa holds similarly.
      \item if the atom is weak, then we have $\ell^i\implies
      p(t_1,\ldots,t_n)\in\Delta_i$, which means $\neg\ell^i\in\Delta_i$ or
      $p(t_1^i,\ldots,t_n^i)\in\Delta_i$. By definition this means
      $s\not\models\ell^i$ or $s\models p(t_1^i,\ldots,t_n^i)$ hence,
      $s\models\ell^i\implies p(t_1^i,\ldots,t_n^i)$, which means $s\models
      L_i(p(t_1,\ldots,t_n))^i$. The \viceversa holds similarly.
      \item if the atom is neither strong nor weak, the thesis follows
      trivially.
    \end{enumerate}
    \item the inductive cases are trivial.
  \end{enumerate}

  Now we can show that for each formula $\psi\in\closure(\phi)$, we have
  $\snf_i(\psi)^i_G\in\Delta_i$ iff $s\models\snf_i(\psi)^i_G$. We do this by
  induction on $\psi$:
  \begin{enumerate}
    \item for the base case of a first-order formula $\lambda$, we have that
      $\snf_i(\lambda)^i_G\equiv L_i(\lambda)^i$, for we showed above,
      $L_i(\lambda)^i\in\Delta_i$ if and only if $s\models L_i(\lambda)^i$,
      which is equivalent to say that $s\models\snf_i(\lambda)^i_G$.
    \item if $\snf_i(\ltl{X\alpha})^i_G\in\Delta_i$, it means
      $(\ltl{X\alpha})^i_G\in\Delta_i$, which by definition of $s$ means
      $s\models(\ltl{X\alpha})^i_G$.
    \item if $\snf_i(\psi_1\lor\psi_2)^i_G\in\Delta_i$, it means
      $\snf_i(\psi_1)^i_G\lor\snf_i(\psi_2)^i_G\in\Delta_i$, which by definition
      of atom holds iff $\snf_i(\psi_1)^i_G\in\Delta_i$ or
      $\snf_i(\psi_2)^i_G\in\Delta_i$. By inductive hypothesis this holds iff
      $s\models\snf_i(\psi_1)^i_G$ or $s\models\snf_i(\psi_2)^i_G$, which is
      equivalent to say that $s\models\snf_i(\psi_1\lor\psi_2)^i_G$. The
      reasoning for the conjunction, the \emph{until} and the \emph{release}
      operators is similar.
  \end{enumerate}
  Finally, we show that this prefix pre-model is for $\phi$, that is,
  $\snf_0(\phi)^0_G\in\Delta_0$. We show by induction that for all $\psi\in\closure(\phi)$,
  if $\psi\in\tlabel^*(\pi_i)$ then $\snf_i(\psi)^i_G\in\Delta_i$:
  \begin{enumerate}
    \item for the base case, for a first-order formula
    $\lambda\in\tlabel(\pi_i)$, we know that $s\models L_i(\lambda)^i$, since it
    is a conjunct of $\Omega(\branch)$ (\emph{w.l.o.g.} we can assume $\psi$ is
    not a conjunction or a disjunction because those are expanded by the
    tableau). But this means $s\models\snf_i(\lambda)^i_G$, and by what we
    showed above, $\snf_i(\lambda)^i_G\in\Delta_i$.
    \item the inductive cases are trivial.
  \end{enumerate}
  This implies that $\snf_0(\phi)^0_G\in\Delta_0$ since
  $\phi\in\tlabel^*(\pi_0)$.
\end{proof}

\begin{lemma}
  \label{lemma:soundness-kunrav-branches-to-models}
  Let $\phi$ be an \LTLfMT formula. If the tableau for $\phi$ contains either an
  accepted branch of $k+1$ poised nodes, or a longer branch, then $\unr{\phi}_k$
  is satisfiable.
\end{lemma}
\begin{proof}
  Let $\branch$ be a branch that is either accepted and has $k+1$ poised nodes
  or that is longer than $k+1$ poised nodes. Let
  $\bar\pi=\seq{\pi_0,\ldots,\pi_k}$ be the sequence of its first $k+1$ poised
  nodes. Since the \trule{Contradiction} rule did not fire on $\pi_{k}$, it
  means that $\Omega(\branch)$ is satisfiable. Let $s=(M,\mu)$ such that
  $s\models\Omega(\branch)$ and let $\bar\Delta=\seq{\Delta_0,\ldots,\Delta_k}$
  be the pre-model for $\phi$ shown to exist by
  \cref{lemma:pre-model-existence}. Recall from the proof of
  \cref{lemma:pre-model-existence} that \emph{w.l.o.g.} we can suppose that
  $s\models(\ltl{X\alpha})^i_G$ if and only if $(\ltl{X\alpha})^i_G\in\Delta_i$.
  
  Now, we can show that $s\models\unr{\phi}_k$. We can do it by induction on
  $k$. If $k=0$, since $\snf_0(\phi)^0_G\in\Delta_0$, by definition of
  $\bar\Delta$ we have $s\models \snf_0(\phi)^0_G$ and so
  $s\models \snf_0(\phi)^0_G$, hence we
  get that $s\models\unr{\phi}_0$. For the inductive step, recall that:
  \begin{align}
    \unr{\phi}_k ={}&\unr{\phi}_{k-1} \land  \ell^{k-1} \land 
    \smashoperator{\bigwedge_{\ltl{X \alpha} \in \XR}} 
    \Big( (\ltl{X \alpha})_G^{k-1} \leftrightarrow \snf_k(\alpha)_G^k \Big)\\
    &\phantom{\unr{\phi}^{k-1}\land  \ell^{k-1} } \land
    \smashoperator{\bigwedge_{\ltl{wX \alpha} \in \wXR}}
    \Big( (\ltl{wX \alpha})_G^{k-1} \leftrightarrow \snf_k(\alpha)_G^k \Big)
  \end{align}
  By the inductive hypothesis, we know that $s\models\unr{\phi}_{k-1}$.
  Moreover, we know that $s\models\ell^{k-1}$, because
  $s\models\Omega(\branch)$. Hence we need to check the two big conjunctions of
  double implications. If $s\models(\ltl{X\alpha})^{k-1}_G$, by definition of
  $s$ we have $(\ltl{X\alpha})^{k-1}_G\in\Delta_{k-1}$, hence by definition of
  prefix pre-model, we have $\snf_k(\alpha)^k_G\in\Delta_k$. By definition of
  $\bar\Delta$, this implies that $s\models\snf_k(\alpha)^k_G$. Conversely, if
  $s\models\snf_k(\alpha)^k_G$, by definition of $\bar\Delta$ we have
  $\snf_k(\alpha)^k_G\in\Delta_k$, which by definition of prefix pre-model means
  that $(\ltl{X\alpha})^{k-1}_G\in\Delta_{k-1}$. By definition of $s$, this
  means $s\models(\ltl{X\alpha})^{k-1}_G$. The reasoning for the weak tomorrows
  is similar. Hence $s$ satisfies the two big conjunctions and then it satisfies
  the whole $\unr{\phi}_k$.
\end{proof}
\begin{lemma}
  \label{lemma:soundness-kunrav-models-to-branches}
  Let $\phi$ be an \LTLfMT formula. If $\unr{\phi}_k$ is satisfiable, then in
  the tableau for $\phi$ there is either an accepted branch, or a branch longer
  than $k+1$ poised nodes.
\end{lemma}
\begin{proof}
  Let $s$ be a $\theory'$-state such that $s\models\unr{\phi}_k$. We use $s$ as a
  guide to navigate the tableau tree to find a suitable branch which is either
  accepted and has $k+1$ poised nodes, or is longer. To do that, we build a
  sequence of branch prefixes $\branch_i=\seq{\node_0,\ldots,\node_i}$ where at
  each step we obtain $\branch_{i+1}$ by choosing $\node_{i+1}$ among the
  children of $\node_i$, until we reach a leaf or $k+1$ poised nodes. During the
  descent, we build a partial function $J:\N\to\N$ that maps positions $j$ in
  $\branch_i$ to indexes $J(j)$ such that, for all $\psi$, it holds that
  $\psi\in\Gamma(\node_j)$ if and only if $s\models\snf_{J(j)}(\psi)_G^{J(j)}$.
  As the base case, we put $\branch_0=\seq{u_0}$ and $J(0)=0$ so that the
  invariant holds since $\Gamma(\node_0)=\set{\phi}$ and
  $s\models\snf_0(\phi)^0_G$ by the definition of $\unr{\phi}^k$. Then,
  depending on the rule that was applied to $\node_i$, we choose $\node_{i+1}$
  among its children as follows.
  \begin{enumerate}
    \item If the \trule{Step} rule has been applied to $\node_i$, then there is
    a unique child that we choose as $\node_{i+1}$, and we define
    $J(i+1)=J(i)+1$. Now, for all $\ltl{X\alpha}\in\tlabel(\node_i)$, we have
    $\alpha\in\tlabel(\node_{i+1})$ by construction of the tableau. Note that
    $\snf_{J(i)}(\ltl{X\alpha})=\ltl{X\alpha}$, hence we know by construction
    that $s\models\ltl{(X\alpha)_G^{J(i)}}$. Then, by definition of
    $\unr{\phi}^k$, we know that $s\models\snf_{J(i)+1}(\alpha)_G^{J(i)+1}$, \ie
    $s\models\snf_{J(i+1)}(\alpha)_G^{J(i+1)}$. On the other direction, if
    $s\models\snf_{J(i+1)}(\alpha)_G^{J(i+1)}$, \ie
    $s\models\snf_{J(i)+1}(\alpha)_G^{J(i)+1}$, then by definition of
    $\unr{\phi}^k$ we have $s\models (\ltl{X\alpha})_G^{J(i)}$, hence
    $s\models\snf_{J(i)}(\ltl{X\alpha})_G^{J(i)}$, and thus
    $\ltl{X\alpha}\in\tlabel(\node_i)$, so by construction of the tableau it
    holds that $\alpha\in\tlabel(\node_{i+1})$. Hence the invariant holds.
    \item If an \emph{expansion rule} has been applied to $\node_i$, then there 
    are one or two children. In both cases, we set $J(i+1)=J(i)$. Then, we 
    proceed as follows.
    \begin{enumerate}
    \item If there is only one child, then it is chosen as $\node_{i+1}$. In
      such a case, the applied rule is necessarily the \trule{Conjunction} one,
      applied to a formula $\psi\equiv\psi_1\land\psi_2$, and thus
      $\psi_1,\psi_2\in\tlabel(\node_{i+1})$. By construction,
      $s\models\snf_{J(i)}(\psi)_G^{J(i)}$, and thus
      $s\models\snf_{J(i+1)}(\psi)_G^{J(i+1)}$. Since
      $\snf_{J(i+1)}(\psi_1\land\psi_2)=\snf_{J(i+1)}(\psi_1)\land\snf_{J(i+1)}(\psi_2)$,
      it holds that $s\models\snf_{J(i+1)}(\psi_1)_G^{J(i+1)}$ and
      $s\models\snf_{J(i+1)}(\psi_2)_G^{J(i+1)}$. For the other direction, if
      $s\models\snf_{J(i+1)}(\psi_1)_G^{J(i+1)}$ and
      $s\models\snf_{J(i+1)}(\psi_2)_G^{J(i+1)}$, then
      $s\models\snf_{J(i+1)}(\psi_1\land\psi_2)_G^{J(i+1)}$, and thus
      $s\models\snf_{J(i)}(\psi_1\land\psi_2)_G^{J(i)}$. Then, by construction,
      it holds that $\psi_1\land\psi_2\in\tlabel(\node_i)$, and thus
      $\psi_1,\psi_2\in\tlabel(\node_{i+1})$. Hence, the invariant holds. 
      
    \item If there are two children $\node_i'$ and $\node_i''$, then let us
      suppose the applied rule is the \trule{Disjunction} rule (similar
      arguments hold for the other rules). In this case, the rule has been
      applied to a formula $\psi\equiv\psi_1\lor\psi_2$, and thus
      $\psi_1\in\tlabel(\node_i')$ and $\psi_2\in\tlabel(\node_i'')$. We know
      that $s\models\snf_{J(i)}(\psi)_G^{J(i)}$, and hence
      $s\models\snf_{J(i+1)}(\psi)_G^{J(i+1)}$. Since
      $\snf_{J(i+1)}(\psi_1\lor\psi_2)=\snf_{J(i+1)}(\psi_1)\lor\snf_{J(i+1)}(\psi_2)$,
      it holds that either $s\models\snf_{J(i+1)}(\psi_1)_G^{J(i+1)}$ or
      $s\models\snf_{J(i+1)}(\psi_2)_G^{J(i+1)}$. Now, we choose $\node_{i+1}$
      accordingly, so to respect the invariant. Note that if both nodes are
      eligible, which one is chosen does not matter. The other direction of the
      invariant holds as well, since if either $s\models\snf_{J(i+1)}(\psi_1)_G^{J(i+1)}$
      or $s\models\snf_{J(i+1)}(\psi_2)_G^{J(i+1)}$, then $s\models\snf_{J(i+1)}(\psi_1)_G^{J(i)}$
      or $s\models\snf_{J(i)}(\psi_2)_G^{J(i)}$, and thus
      $s\models\snf_{J(i)}(\psi_1\lor\psi_2)_G^{J(i)}$. Hence,
      $\psi_1\lor\psi_2\in\tlabel(\node_i)$, and then either
      $\psi_1\in\tlabel(\node_{i+1})$ or $\psi_2\in\tlabel(\node_{i+1})$.
    \end{enumerate}
  \end{enumerate}
  Let $\branch=\seq{\node_0,\ldots,\node_i}$ be the branch prefix built as
  explained above, and let $\bar\pi=\seq{\pi_0,\ldots,\pi_n}$ be the sequence of
  its poised nodes. As already pointed out, the descent stops when $\pi_n$ is a
  leaf or when $n=k$. Note that, in any case, $\node_i=\pi_n$. In case we find a
  leaf, it is not possible that it has been crossed by the \trule{Contradiction}
  rule. We show that by showing that $\Omega(\branch)$ is satisfiable and in
  particular, that $s\models\Omega(\branch)$. This can be seen by considering
  that $\Omega(\branch)$ includes the conjunction of $(L_i(\psi))^i$ for each
  first-order formula $\psi\in\tlabel(\pi_i)$. Note as well that
  $\snf_i(\psi)^i_G\equiv (L_i(\psi))^i$ for first-order formulas. By the
  invariant maintained during the descent, $s\models\snf_i(\psi)^i_G$, hence
  $s\models(L(\psi))^i$. Note that since $s\models\unr{\phi}_k$, in particular
  it holds that $s\models\ell^i$ for $0\le i < k$. Altogether, we have that
  $s\models\Omega(\branch)$. Therefore, $\branch$ cannot have been rejected, and
  is thus accepted or the prefix of a branch longer than $k+1$ poised nodes.
\end{proof}

\begin{lemma}
  \label{lemma:empty-encoding-left-to-right}
  Let $\phi$ be an \LTLfMT formula. If the tableau for $\phi$ contains an
  accepted branch of $k+1$ poised nodes, then $\encod{\phi}_k$ is satisfiable.
\end{lemma}
\begin{proof}
  Suppose that the complete tableau for $\phi$ contains an accepted branch of
  $k+1$ poised nodes and no shorter accepted branches, so let
  $\branch=\seq{\node_0,\ldots,\node_n}$ be such a branch, and let
  $\bar\pi=\seq{\pi_0,\ldots,\pi_k}$ be the sequence of its poised nodes. By
  \cref{lemma:soundness-kunrav-branches-to-models} we know that $\unr{\phi}_k$
  is satisfiable, hence we need to show that $s\not\models(\ltl{X\alpha})$ for
  all $\ltl{X\alpha}\in\XR$ and $s\models\neg\ell^k$. Since the branch is
  accepted, we know there is a model $s=(M,\mu)$ such that
  $s\models\Omega(\branch)\land\neg\ell^k$. Thus it follows directly that
  $s\models\neg\ell^k$. Now, let $\bar\Delta=\seq{\Delta_0,\ldots,\Delta_k}$ be
  the pre-model found as in \cref{lemma:pre-model-existence}, which has the
  property that $s\models\snf_i(\psi)^i_G$ iff $\snf_i(\psi)^i_G\in\Delta_i$.
  Note, as noted in the proof of \cref{lemma:pre-model-existence}, that
  \emph{w.l.o.g.} we can assume that $s$ is such that
  $s\models(\ltl{X\alpha})^i_G$ if and only if
  $(\ltl{X\alpha})^i_G\in\Delta_i$, for all $0\le i \le k$. By
  \cref{thm:soundness}, we have a state sequence $\sigma$ such that 
  $\sigma\models\phi$. Recall that this model
  $\sigma=\seq{(M,\mu_0),\ldots,(M,\mu_k)}$ is defined such that
  $\mu_i(x)=\mu(x^i)$. Recall that this implies that
  $\eval{t}^i_{\sigma,\xi}=\eval{t^i}_{s,\xi}$. Thanks to that, by induction we
  can show that for all first-order formulas $\lambda\in\closure(\phi)$, it
  holds that $\sigma,\xi,i\models \lambda$ iff $s,\xi\models L_i(\lambda)^i$.
  For the base case, we distinguish whether $p(t_1,\ldots,t_n)$ is
  \emph{strong}, \emph{weak}, or neither:
  \begin{enumerate}
    \item if the atom is \emph{strong}, we have that $L_i(p(t_1,\ldots,t_n))^i$
      is $\ell^i\land p(t_1^i,\ldots,t_n^i)$. If $\sigma,\xi,i\models
      p(t_1,\ldots,t_n)$, it means
      $p(\eval{t_1}^i_{\sigma,\xi},\ldots,\eval{t_n}^i_{\sigma,\xi})\in p^M$,
      which means $p(\eval{t_1^i}_{s,\xi},\ldots,\eval{t_n^i}_{s,\xi})\in p^M$,
      which means $s,\xi\models p(t_1^i,\ldots,t_n^i)$. Since the atom is
      satisfied and it is strong, we know $i<k$, hence, since $s$ is a model of
      $\Omega(\branch)$, it means $s,\xi\models\ell^i$. Hence $s,\xi\models
      L_i(p(t_1,\ldots,t_n))^i$. Conversely, if $s,\xi\models
      L_i(p(t_1,\ldots,t_n))$, we know $i<k$ because $s,\xi\models\ell^i$. Then,
      the strong atom is well formed, and we know that $s,\xi\models
      p(t_1^i,\ldots,t_n^i)$ means
      $p(\eval{t_1^i}_{s,\xi},\ldots,\eval{t_n^i}_{s,\xi})\in p^M$, which is
      equivalent to saying that
      $p(\eval{t_1}^i_{\sigma,\xi},\ldots,\eval{t_n}^i_{\sigma,\xi})\in p^M$,
      hence $\sigma,\xi,i\models p(t_1,\ldots,t_n)$.
    \item if the atom is \emph{weak}, we have that $L_i(p(t_1,\ldots,t_n))^i$ is
      $\ell^i\implies p(t_1^i,\ldots,t_n^i)$. If $\sigma,\xi,i\models
      p(t_1,\ldots,t_n)$ it means that either $i=k$, in which case $s$ satisfies
      the implication because $s\not\models\ell^i$, or $i<k$ and the atom is
      well formed, hence
      $p(\eval{t_1}^i_{\sigma,\xi},\ldots,\eval{t_n}^i_{\sigma,\xi})\in p^M$,
      which means $p(\eval{t_1^i}_{s,\xi},\ldots,\eval{t_n^i}_{s,\xi})\in p^M$,
      which is $s,\xi\models p(t_1^i,\ldots,t_n^i)$. Conversely, if
      $s,\xi\models L_i(p(t_1,\ldots,t_n))^i$, either $s\not\models\ell^i$,
      which happens only if $i=k$ hence $\sigma,\xi,i\models p(t_1,\ldots,t_n)$
      by definition, or the atom is well formed and we have
      $p(\eval{t_1^i}_{s,\xi},\ldots,\eval{t_n^i}_{s,\xi})\in p^M$ which is
      $p(\eval{t_1}^i_{\sigma,\xi},\ldots,\eval{t_n}^i_{\sigma,\xi})\in p^M$,
      hence $\sigma,\xi,i\models p(t_1,\ldots,t_n)$.
    \item if the atom is neither strong nor weak, the thesis holds trivially.
  \end{enumerate}
  As for the inductive cases, they all follow trivially from the definition of
  $L_i(\lambda)^i$. Now, by induction we can show that for all
  $\psi\in\closure(\phi)$, it holds that $\sigma,i\models\psi$ iff
  $s\models\snf_i(\psi)^i_G$. We already proved above the base case. As for the
  inductive cases, they follow trivially from the definition of
  $\snf_i(\psi)^i_G$. We write the case for tomorrow formulas. If
  $\sigma,i\models\ltl{X\alpha}$, it means $i<k$ and $\sigma,i+1\models\alpha$,
  which by the inductive hypothesis means $s\models\snf_{i+1}(\alpha)^{i+1}_G$.
  By definition of $\bar\Delta$, we have that
  $\snf_{i+1}(\alpha)^{i+1}_G\in\Delta_{i+1}$, which by definition of prefix
  pre-models, means $(\ltl{X\alpha})^i_G\in\Delta_i$. By definition of $s$, we
  now have that $s\models(\ltl{X\alpha})^i_G$, which is
  $s\models\snf_i(\ltl{X\alpha})^i_G$. The converse direction is similar. So we
  proved that for all $\psi\in\closure(\phi)$, it holds that
  $\sigma,i\models\psi$ iff $s\models\snf_i(\psi)^i_G$. But this means that
  $\sigma,i\models\psi$ iff $\snf_i(\psi)^i_G\in\Delta_i$. Hence, at step $k$,
  since $\sigma$ is a model of $\phi$ and we know by definition that
  $\sigma,k\not\models\ltl{X\alpha}$ for any $\ltl{X\alpha}\in\closure(\phi)$,
  it follows that $\ltl{X\alpha}\not\in\Delta_k$, hence
  $s\not\models(\ltl{X\alpha})^i_G$.
\end{proof}

\begin{lemma}
  \label{lemma:empty-encoding-right-to-left}
  Let $\phi$ be an \LTLfMT formula. If $\encod{\phi}_k$ is satisfiable, then the
  tableau for $\phi$ contains an accepted branch.
\end{lemma}
\begin{proof}
  Suppose that $\encod{\phi}_k$ is satisfiable. Hence, there exists a model $s$
  such that $s\models\encod{\phi}_k$. Then, $\unr{\phi}_k$ is satisfiable, and
  we know from \cref{lemma:soundness-kunrav-models-to-branches} that the
  complete tableau for $\phi$ has either an accepted branch, or a branch longer
  than $k+1$ poised nodes. Let $\branch=\seq{\node_0,\ldots,\node_n}$ be the
  branch prefix found as shown in the proof of
  \cref{lemma:soundness-kunrav-models-to-branches}, and let
  $\bar\pi=\seq{\pi_0,\ldots,\pi_k}$ be the sequence of its poised nodes. By
  construction, there exists a function $J:\N\to\N$ fulfilling the invariant:
  $\psi\in\tlabel(u_i)$ if and only if $s\models\snf_{J(i)}(\psi)_G^{J(i)}$. We
  now show that indeed $\branch$ is accepted. Since $\encod{\phi}_k$ is
  satisfiable, then $s\models\neg\psi_G^k$ for each $\psi\in\XR$. Since $\psi$
  is an $\ltl{X}$-request, $\snf_k(\psi)\equiv\psi$, and thus
  $s\not\models\snf_k(\psi)_G^k$. Here, $k=J(j)$, for some $j$, and, from the
  invariant, it follows that $\psi\not\in\tlabel(\node_j)$. Hence, $\node_j$
  does not contain any $\ltl{X}$-request. Furthermore, note that
  $s\models\Omega(\branch)$ because the branch is not rejected by the
  \trule{Contradiction} rule, and $s\models\neg\ell^k$ as stated by
  $\encod{\phi}_k$. Hence, the \trule{Empty} rule is triggered, accepting the
  branch. 
\end{proof}

Now we can prove the following.
\encodingthm*
\begin{proof}
  For the left-to-right direction, suppose the procedure returns \SAT. Then it
  means $\encod{\phi}_k$ is satisfiable for some $k\ge0$, which by
  \cref{lemma:empty-encoding-right-to-left} implies that the complete tableau
  for $\phi$ contains an accepted branch.

  For the right-to-left direction, suppose that the tableau for $\phi$ contains
  an accepted branch $\branch$ of $k+1$ poised nodes for some $k\ge0$. Thus, for
  any $j<k$, there is a branch longer than $j+1$ poised nodes, hence by
  \cref{lemma:soundness-kunrav-branches-to-models} $\unr{\phi}_j$ is
  satisfiable, and Line~5 does not return \UNSAT at the $j$-th iteration of the
  algorithm. Now, at the $k$-th iteration, by
  \cref{lemma:soundness-kunrav-branches-to-models} we know Line~5 again does not
  return \UNSAT. Moreover, by \cref{lemma:empty-encoding-left-to-right}, the
  existence of $\branch$ implies that $\encod{\phi}_k$ is satisfiable, hence
  Line 8 returns \SAT.
\end{proof}

\fi

\end{document}